\documentclass[preprint,11pt]{elsarticle}
\makeatletter

\let\ps@pprintTitle\ps@plain

\makeatother
 \journal{ }
\usepackage[dvipsnames, svgnames, x11names]{xcolor}
\usepackage{graphicx}
\usepackage{pifont,latexsym,ifthen,amsthm,rotating,calc,textcase,booktabs}
\usepackage{amsfonts,amssymb,amsbsy,amsmath,extarrows,bm}
\newtheorem{theorem}{Theorem}[section]
\newtheorem{lemma}[theorem]{Lemma}

\newtheorem{remark}[theorem]{Remark}
\newtheorem{proposition}[theorem]{Proposition}

\newtheorem{prob}{RH problem}[section]
\newtheorem{prob2}{$\bar{\partial}$-RH problem}[section]
\newtheorem{prob3}{$\bar{\partial}$-problem}[section]
\numberwithin{equation}{section}
\usepackage{natbib}
\usepackage{tikz}
\usepackage{stmaryrd}
\usepackage[T1]{fontenc}
\usepackage[english]{babel}
\usepackage{listings}
\usepackage{xcolor,mathrsfs,url}
\usepackage{ifthen}
\usepackage{fancyhdr}
\usepackage{verbatim}
\usepackage{todonotes}
\usepackage{float}
\usepackage{subfigure}
\usepackage{hyperref}
\usepackage{appendix}
\usepackage[percent]{overpic}
\usepackage {caption}
\hypersetup{hypertex=true,
            colorlinks=true,
            linkcolor=blue,
            anchorcolor=blue,
            citecolor=red}

\DeclareMathOperator*{\im}{Im}
\DeclareMathOperator*{\re}{Re}
\DeclareFontFamily{U}{mathx}{}
\DeclareFontShape{U}{mathx}{m}{n}{<-> mathx10}{}
\DeclareSymbolFont{mathx}{U}{mathx}{m}{n}
\DeclareMathAccent{\widehat}{0}{mathx}{"70}
\DeclareMathAccent{\widecheck}{0}{mathx}{"71}

\begin{document}

\begin{frontmatter}
\title{Painlev\'e   transcendents in the  defocusing mKdV equation  with  non-zero boundary conditions}

\author[inst2]{Zhaoyu Wang$^{\dag,}$}

\author[inst2]{Taiyang Xu$^{\ast,}$}

\author[inst2]{Engui Fan$^{\star,}$  }

\address[inst2]{ School of Mathematical Sciences and Key Laboratory of Mathematics for Nonlinear Science, Fudan University, Shanghai, 200433, China\\
$\dag$ E-mail address: wang\_zy@fudan.edu.cn;\ \ \
* E-mail\, address: tyxu19@fudan.edu.cn;\\
$\star  $Corresponding author and e-mail address: faneg@fudan.edu.cn  }

\begin{abstract}
We consider the Cauchy problem for the  defocusing modified Korteweg-de Vries (mKdV) equation  with non-zero boundary conditions
\begin{align}
    &q_t(x,t)-6q^2(x,t)q_{x}(x,t)+q_{xxx}(x,t)=0,    \nonumber\\
    &q(x,0)=q_{0}(x)\to \pm 1, \ \ x\rightarrow\pm\infty,\nonumber
\end{align}
which can be   characterized using a   Riemann-Hilbert  problem through the inverse scattering transform.  Using the $\bar\partial$-generalization of the Deift-Zhou nonlinear steepest descent approach, combined with the double scaling limit  technique,  we obtain the long-time   asymptotics  of the solution of the Cauchy problem for the defocusing
mKdV equation in the transition region $|x/t+6|t^{2/3}< C$ with $C>0$.  The asymptotics
 can be expressed in terms of the solution of the second  Painlev\'{e}  transcendent.
\end{abstract}

\begin{keyword}
defocusing mKdV equation, \  Riemann-Hilbert problem, \ $\bar\partial$-steepest descent method, \ Painlev\'{e}  transcendents, \ long-time asymptotics

\textit{Mathematics Subject Classification:} 35P25; 35Q51; 35Q15; 35B40; 35C20.
\end{keyword}

\end{frontmatter}

\tableofcontents

\section{Introduction }

This  paper is concerned  with     the  Painlev\'e  asymptotics  of the defocusing  modified Korteweg-de Vries (mKdV)  equation with non-zero boundary conditions
\begin{align}
    &q_t(x,t)-6q^2(x,t)q_{x}(x,t)+q_{xxx}(x,t)=0,\label{dmkdv}\\
    &q(x,0)=q_{0}(x)\to \pm1, \ \ x\rightarrow\pm\infty,\label{bdries}
\end{align}
where $q_0(x) -\tanh{(x)} \in H^{4,4}(\mathbb{R})$.
The mKdV equation     arises in various  physical fields, such as acoustic wave and phonons in a certain anharmonic
lattice \cite{Zab1967,Ono1992}, as well as Alfv\'en wave in a cold collision-free plasma \cite{kaku1969, Kha1998}.
  The  mKdV equation on the line is locally well-posed \cite{mkdv1} and globally well-posed in $H^s(\mathbb{R})$ for $s \ge \frac{1}{4}$ \cite{mkdv2, mkdv3, mkdv4}.
   Recently,  the global  well-posedness to the Cauchy problem  for the mKdV equation was further generalized to the space   $H^{s}(\mathbb{R})$ for $s >-\frac{1}{2}$
   \cite{mkdv6}.

It is  well-established that
   the defocusing mKdV
equation \eqref{dmkdv} with zero boundary conditions (ZBCs, i.e., $q_0(x)\rightarrow 0$ as $x\rightarrow\pm\infty$), does not exhibit solitons due to the absence of discrete spectrum in the
 self-adjoint ZS-AKNS scattering operator    (see (\ref{lax pair}) below) \cite{DeiftP2}.
Numerous studies have been conducted to analyze the long-time asymptotic behavior of solutions to   the equation \eqref{dmkdv} in the continuous spectrum without   solitons.
The earliest work can be traced back to Segur and Ablowitz \cite{AblwzP1}, who derived  the leading asymptotics of the solutions of the mKdV and Korteweg-de Vries (KdV) equations, including  the full information on the phase.
The Deift-Zhou nonlinear steepest descent method \cite{DeiftP2} has significantly influenced research on the long-time  behavior of the mKdV equation (\ref{dmkdv}), rigorously deriving the asymptotics for all relevant regions, including the self-similar Painlev\'e region, under the conditions of soliton free.
The long-time asymptotic behavior of the solution  to the mKdV equation  with step-like initial data has been extensively studied in previous works \cite{MK0,MK1,GM,VA1,VA1ST}.
Boutet de Monvel \emph{et al.}  discussed the initial boundary value problem of the defocusing mKdV equation in the finite interval
 using the Fokas method \cite{Monvel}.
Moreover, the long-time asymptotics  of  the solution to the defocusing
mKdV  equation (\ref{dmkdv})  was established for initial data   in  a  weighted Sobolev space  $H^{2,2}(\mathbb{R})$ without considering solitons  \cite{chenliu}.
Furthermore, Charlier and Lenells investigated the Airy and Painlev\'e asymptotics for the mKdV equation  \cite{Charlier2020}, and later, Huang and Zhang extended these asymptotics to the entire mKdV hierarchy \cite{Huang2020}.
The Painlev\'e asymptotics in transition regions also appear in other integrable systems.
Segur and Ablowitz described the asymptotics in the transition region for the  KdV  equation \cite{AblwzP1}. The connection  between    the  tau-function of the sine-Gordon  reduction and the Painlev\'e \uppercase\expandafter{\romannumeral3} equation was established through  the RH  approach \cite{Its3}.
Additionally, Boutet de Monvel \emph{et al.} obtained  the Painlev\'e asymptotics for the Camassa-Holm equation by using the nonlinear steepest descent approach \cite{Monvel2010}.
More recently, we  found  the Painlev\'e asymptotics for the defocusing nonlinear Schr\"odinger (NLS) equation with non-zero boundary conditions (NZBCs, i.e., $q_0(x)\nrightarrow 0$ as $x\rightarrow\pm\infty$)  \cite{wfp}.
Moreover, it also appears in the modified  Camassa-Holm equation \cite{xyz}.

However, it is worth noting that the defocusing mKdV equation \eqref{dmkdv} with  NZBCs  \eqref{bdries}
allows the existence of  solitons   due to the presence of  a non-empty  discrete spectrum. The corresponding
 $N$-soliton solutions  were skillfully constructed using the inverse scattering transform (IST)  \cite{Z&Y}.
  Recently,  by  using the $\bar\partial$ steepest descent method, which  was  introduced in  \cite{McL1,McL2} and has been extensively implemented
in  the long-time asymptotic analysis  and  soliton resolution conjecture of   integrable systems  \cite{fNLS,Liu3,LJQ,YF1,WF},
the  long-time asymptotics  of the solution to  the Cauchy problem \eqref{dmkdv}-\eqref{bdries} was obtained in three different  regions:  a solitonic region $-6<\xi\le-2$ (where $\xi:=x/t$) \cite{zx1} and
two solitonless region $ \xi<-6$ and $  \xi>-2$ \cite{zx2} (see Figure \ref{cone}).
The remaining question is: How to describe   the  asymptotics of the solution to  the Cauchy problem \eqref{dmkdv}-\eqref{bdries}
in the transition region $\xi \approx  -6$ ?

In this paper,
  we demonstrate that  the long-time asymptotics   of the solution to the Cauchy problem \eqref{dmkdv}-\eqref{bdries} in this transition region   can be expressed in terms  of    the solution of the  Painlev\'e \uppercase\expandafter{\romannumeral2} equation.
In the generic case
of the mKdV equation,
   the  norm  $ (1- |r(z)|^2)^{-1}$    blows up as $z  \to   \pm 1$.  This is not merely a technical difficulty,
   but indicates the    emergence of a new phenomenon that
   cannot be treated in  the same manner as the   cases in  \cite{zx1,zx2}. In the context of our research, we confirm   the Painlev\'e asymptotics  in the  transition region $\xi \approx  -6$.
Compared  with  the case of ZBCs \cite{DeiftP2,Charlier2020,Huang2020},   the case of  NZBCs  we considered    meets substantial difficulties.
    Firstly, due to the effect of solitons on the Cauchy problem \eqref{dmkdv}-\eqref{bdries}, a more detailed description is necessary to formulate a solvable model.
 Secondly, in the case of  the mKdV equation \eqref{dmkdv}   with  ZBCs,    the phase function is given by
$$\theta(z)=2z^3+zx/t, $$
and the corresponding  RH problem  can directly   match   the    Painlev\'{e} \uppercase\expandafter{\romannumeral2} model RH problem (see  \ref{appx})  \cite{DeiftP2,Charlier2020}.
However, in  the   case  of NZBCs  \eqref{dmkdv}-\eqref{bdries},   the phase function becomes
$$\theta(z)= \frac{1}{2}(z-z^{-1})\left[x/t+(z+z^{-1})^2 +2\right], $$
whose  corresponding  RH problem cannot  directly   match   a solvable     Painlev\'{e}  RH problem.
To confront this difficulty,  we propose a key technique to approximate the phase function of  this  RH problem to that of  the Painlev\'{e} \uppercase\expandafter{\romannumeral2} model RH problem using double series (see (\ref{tthe1})-(\ref{tthea}) below).
 By doing so, we find the Painlev\'e  asymptotics for the mKdV equation under   NZBCs in the transition region $|x/t+6|t^{2/3}< C$ with $C>0$.

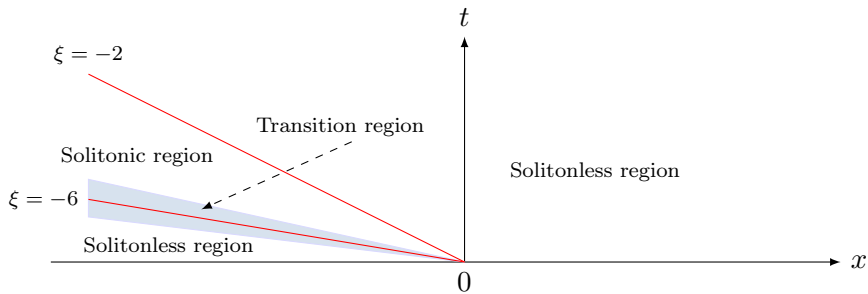
\begin{figure}[htbp]
    \begin{center}
    \begin{tikzpicture}[node distance=2cm]
    \draw[-latex](-5.5,0)--(5,0)node[right]{$x$};
    \draw[-latex](0,0)--(0,3)node[above]{$t$};
     \draw[blue!15, fill=LightSteelBlue!50] (0,0)--(-5,0.6)--(-5,1.1)--(0,0);
    \draw[ red](0,0)--(-5,5/6)node[left,black]{\scriptsize{$\xi=-6 $}};
    \draw[ red](0,0)--( -5,2.5)node[above,black]{\scriptsize{$\xi=-2 $}};

    \node[below]{$0$};
    \coordinate (A) at (-5.2, 0.2);
	\fill (A) node[right] {\scriptsize{Solitonless region}};
    \coordinate (C) at (-5.5, 1.4);
    \fill (C) node[right] {\scriptsize{Solitonic region}};
     \coordinate (D) at (3, 1.2);
	 \fill (D) node[left] {\scriptsize{Solitonless region}};
    \coordinate (E) at (-0.4, 1.8);
	 \fill (E) node[left] {\scriptsize{Transition region}};
       \draw[dashed,-latex](-1.5,1.6)--(-3.5,0.7);
    \end{tikzpicture}
    \caption{\small  The $(x,t)$-plane is divided into  three kinds of asymptotic regions:
    Solitonic region, $-6 <\xi\le-2 $; Solitonless region,  $\xi<-6 $ and  $\xi>-2 $;  Transition region, $\xi\approx -6 $.} \label{cone}
    \end{center}
    \end{figure}

The organization of the present paper is as follows:
In Section \ref{sec1}, we  consider the forward scattering transform for the Cauchy problem (\ref{dmkdv})-(\ref{bdries}),
 including  the properties of the Jost functions and the scattering data derived from the initial data.
In Section \ref{rh},  we perform the inverse scattering transform and establish a  matrix-valued RH problem associated with  this Cauchy problem.  Furthermore,
 we  transform the original RH problem to a regular RH problem, removing the effect of solitons and the spectral singularities.
In Section \ref{sec3},  we investigate the Painlev\'e  asymptotics in the transition region
$\left|x/t +6 \right|t^{2/3}<C$ for any $C>0$.
The  $\bar\partial$-steepest descent method and  the double scaling limit technique    are applied to deform  the regular RH problem  into
a solvable RH problem that matches with  the  Painlev\'{e} \uppercase\expandafter{\romannumeral2} model RH problem.
Following the analysis presented in the preceding sections, we establish   the main result on the Painlev\'e asymptotics for the defocusing mKdV equation in Theorem \ref{th}.


\begin{remark}
 For the generic case when $  |r(\pm 1)|=1 $,
Cuccagna and Jenkins proposed a new way in \cite{CJ}  to get rid of  the  restrictive condition $\|r\|_{L^{\infty}(\mathbb{R})}<1$
 by specially  handling the singularity caused by   $ |r(\pm 1) | = 1$.
 By using this method, the asymptotic properties of the similar and self-similar region can be matched,
 and thus  no new shock wave asymptotic  forms  appear.
 Because of this,  the  Painlev\'e asymptotics given   in Theorem \ref{th} are still effective when $|r(\pm 1)|=1$ generically.
\end{remark}

\paragraph{\bf Notations}
We introduce some notations that  will  be used  in this paper:
\begin{itemize}
    \item  $L^{p,s}(\mathbb{R})$ defined with the norm $\|q\|_{L^{p,s}(\mathbb{R})} := \|\langle x \rangle^s q\|_{L^{p}(\mathbb{R}) }$,  where $\langle x \rangle=\sqrt{1+x^2}$.
    \item $W^{k,p}(\mathbb{R})$ defined with the norm $\|q\|_{W^{k,p}(\mathbb{R})}:=\sum_{j=0}^k \|\partial^jq\|_{L^p(\mathbb{R})}$, where $\partial^jq$ is the $j^{th}$ weak derivative of $q$.
   \item$H^k(\mathbb{R})$ defined with the norm $\|q\|_{H^k(\mathbb{R})}:=\|\langle x \rangle^k \widehat q\|_{L^2(\mathbb{R})}$, where $\widehat q$ is the Fourier transform of $q$.
  \item $H^{k,s} (\mathbb{R}) : = H^k(\mathbb{R}) \cap L^{2,s}(\mathbb{R})$ defined with the norm
    $\|q\|_{H^{k, s}(\mathbb{R})}  :=     \|  q \|_{H^k(\mathbb{R})}+\|\langle x \rangle^s q\|_{L^{2}(\mathbb{R}) }.$
    \item As usual, the three Pauli matrices are defined by
    \begin{equation}
    \sigma_1 = \left(\begin{array}{cc}
  0 & 1\\
  1 & 0\\
 \end{array}\right), \quad \sigma_2 = \left(\begin{array}{cc}
  0 & -i\\
  i & 0\\
 \end{array}\right), \quad \sigma_3 = \left(\begin{array}{cc}
  1 & 0\\
  0 & -1\\
 \end{array}\right). \nonumber
    \end{equation}
    \item For a complex-valued function $f(z)$ where $z \in \mathbb{C}$, we use
   $  f^*(z):= \overline{f(\bar{z})}$
  to denote the Schwarz conjugation.
    \item $\hat{\sigma}_3 $ acts on a matrix $A$ by   $e^{\hat{\sigma}_3 }A = e^{\sigma_3}  A e^{-\sigma_3} $.
\end{itemize}

\section{Forward  Scattering Transform} \label{sec1}
In this section, we review main results about the forward scattering transform of the defocusing mKdV equation with  weighted Sobolev initial data. A comprehensive exposition of these results can be found in \cite{zx1}.

\subsection{Jost functions }
The Lax pair of the defocusing mKdV equation \eqref{dmkdv} is given by
\begin{equation}\label{lax pair}
    \Phi_x=X\Phi, \quad \Phi_t=T\Phi,
\end{equation}
where
\begin{align*}
    &X=ik\sigma_3+Q,\ \ T=4k^2X-2ik\sigma_3(Q_x-Q^2)+2Q^3-Q_{xx},
\end{align*}
$k \in \mathbb{C}$ is a spectral parameter, and
$$Q=Q(x,t)=\left(\begin{array}{cc}
  0 & q(x,t)\\
   q(x,t) & 0\\
 \end{array}\right).$$

Under the boundary condition (\ref{bdries}), we then get the asymptotic
spectral problem
\begin{equation}\label{asy spec prob}
    \phi_{\pm, x}= X_{\pm}\phi_{\pm}, \quad
    \phi_{\pm, t}= T_{\pm}\phi_{\pm}, \quad  x\rightarrow\pm\infty,
\end{equation}
where 
\begin{equation*}
    X_{\pm}=ik\sigma_3+Q_{\pm}, \quad T_{\pm}=(4k^2+2)X_{\pm}, \quad  Q_{\pm} = \pm \sigma_1.
\end{equation*}
The eigenvalues of $X_{\pm}$ are $\pm i\lambda$, which satisfy the equality
\begin{equation}
    \lambda^2=k^2-1.
\end{equation}
Since the eigenvalue $\lambda$ is multi-valued, we introduce the following uniformization variable
\begin{equation}
    z=k+\lambda,
\end{equation}
and obtain two single-valued functions
\begin{equation}
    \lambda(z)=\frac{1}{2}(z-z^{-1}),\quad k(z)=\frac{1}{2}(z+z^{-1}).
\end{equation}
We derive from the asymptotic spectral problem \eqref{asy spec prob} that
\begin{equation}
    \phi_{\pm}(z)= E_{\pm}(z)e^{ i\lambda(z)x\sigma_3},
\end{equation}
where
\begin{equation*}
    E_{\pm}(z)=I\mp z^{-1}\sigma_2.
\end{equation*}
As usual, we define the Jost functions $\Phi_{\pm}$ such that
\begin{equation*}
\Phi_{\pm}(z)\sim E_{\pm}(z)e^{i\lambda(z)x\sigma_3}, \quad {\rm as} \ \ x\to\pm\infty.
\end{equation*}
Subsequently, the modified Jost functions are defined by
\begin{equation}\label{muPhirelation}
    \mu_{\pm}(z)=\Phi_{\pm}(z)e^{-i\lambda(z)x\sigma_3},
\end{equation}
and we then have
\begin{align*}
    &\mu_{\pm}(z)\sim E_{\pm}(z), \quad {\rm as} \ x\rightarrow\pm\infty, \\
    &{\rm det}(\Phi_{\pm}(z))={\rm det}(\mu_{\pm}(z))={\rm det}(E_{\pm}(z))=1-z^{-2}.
\end{align*}
Furthermore, $\mu_{\pm}(z)$ could be defined by the Volterra type integral equations
\begin{align}
    &\mu_{\pm}(z)=E_{\pm}(z)+\int_{\pm\infty}^{x} E_{\pm}(z)e^{i\lambda(z)(x-y)\hat{\sigma}_3}\left(E^{-1}_{\pm}(z)\Delta Q_{\pm}\left(y\right)\mu_{\pm}\left( z;y\right)\right)\mathrm{d}y, \quad z\neq \pm 1,\label{Vo1}\\
    &\mu_{\pm}(z)=E_{\pm}(z)+\int_{\pm\infty}^{x} \left(I+\left(x-y\right)\left(Q_{\pm}\pm i\sigma_3\right)\right)\Delta Q_{\pm}\left(y\right)\mu_{\pm}\left(z;y\right)\mathrm{d}y, \quad z=\pm 1,\label{Vo2}
\end{align}
where $\Delta Q_{\pm}=Q-Q_{\pm}$.

Denote
$\mathbb{C}^\pm =\{z \in \mathbb{C}:\pm \im z >0\}$.
 Let $\mu_{\pm}(z)=\left( \mu_{\pm,1}(z),  \mu_{\pm,2}(z)\right)$.  The properties of $\mu_{\pm} (z)$  are concluded in the following lemma  \cite{zx1}.
\begin{lemma}\label{lsym1}
    	 Given $n\in\mathbb{N}_0$, let $q_0(x)-\tanh{(x)}\in L^{1,n+1}(\mathbb{R})$ and $q'_0(x) \in W^{1,1}(\mathbb{R})$. Then

\begin{itemize}

 \item   { Analyticity:}  For $z\in\mathbb{C}\setminus\{0\}$, $\mu_{+,1}(z)$ and $\mu_{-,2}(z)$ can be analytically extended to $\mathbb{C}^+$ and continuously extended to $\mathbb{C}^+\cup \mathbb{R}$;
    $\mu_{-,1}(z)$ and $\mu_{+,2}(z)$ can be analytically extended to $\mathbb{C}^-$ and continuously extended to $\mathbb{C}^-\cup \mathbb{R}$.

    \item {Symmetry:} $\mu_{\pm}(z)$ satisfies the symmetries
         	\begin{equation} \label{symmu}
    	\mu_{\pm}(z)=\sigma_1\overline{\mu_{\pm}(\bar{z})}\sigma_1=\overline{\mu_{\pm}(-\bar{z})}=\mp z^{-1}\mu_{\pm}\left(z^{-1}\right)\sigma_2.
    	\end{equation}
  \item    { Asymptotic behavior as $z\to \infty$:  } For $\im z \ge 0$, as $z\to\infty$,
    	    \begin{align*}
    &\mu_{+,1}(z)=e_1+z^{-1}\left(\begin{array}{c}
                                    -i\int_{x}^{\infty}(q^2-1)\mathrm{d}x\\
                                    -iq
                                    \end{array}\right)+\mathcal{O}\left(z^{-2}\right),\\
    &\mu_{-,2}(z)=e_2+z^{-1}\left(\begin{array}{c}
                                    iq\\
                                    i\int^{x}_{-\infty}(q^2-1)\mathrm{d}x
                                    \end{array}\right)+\mathcal{O}\left(z^{-2}\right);
    \end{align*}
    For $\im z \le 0$, as $z\to\infty$,
     \begin{align*}
    &\mu_{-,1}(z)=e_1+z^{-1}\left(\begin{array}{c}
                                    -i\int^{x}_{-\infty}(q^2-1)\mathrm{d}x\\
                                    -iq
                                    \end{array}\right)+\mathcal{O}\left(z^{-2}\right),\\
    &\mu_{+,2}(z)=e_2+z^{-1}\left(\begin{array}{c}
                                    iq\\
                                    i\int_{x}^{\infty}(q^2-1)\mathrm{d}x
                                    \end{array}\right)+\mathcal{O}\left(z^{-2}\right),
    \end{align*}
    where $e_1 = (1,0)^\textnormal{T}$ and $e_2 = (0,1)^\textnormal{T}$.
  \item    { Asymptotic behavior as $z \to 0$:} For $z \in \mathbb{C}^+$, as $z\to 0$,
     \begin{equation*}
    \mu_{+,1}(z)=-i z^{-1}e_2+\mathcal{O}(1),\hspace{0.5cm}\mu_{-,2}(z)=-i z^{-1}e_1+\mathcal{O}(1);
    \end{equation*}
    For $z \in \mathbb{C}^-$, as $z \to 0$,
        \begin{equation*}
    \mu_{-,1}(z)=iz^{-1}e_2+\mathcal{O}(1),\hspace{0.5cm}\mu_{+,2}(z)=iz^{-1}e_1+\mathcal{O}(1).
    \end{equation*}
  \end{itemize}

    \end{lemma}

\subsection{Scattering data }
The Jost functions   $\Phi_{\pm}(z)$  satisfy the linear relation
 \begin{equation}
    \Phi_{+}(z)=\Phi_{-}(z)S(z),
    \end{equation}
where $S(z)$ is the scattering matrix given by
    $$S(z)=\begin{pmatrix} a(z) & \overline{b(z)} \\ b(z) & \overline{a(z)} \end{pmatrix},\quad z\in \mathbb{R}\setminus \{0, \pm 1\},$$
   and $a(z)$ and $b(z)$ are the scattering coefficients, by which   we define the reflection coefficient
	    \begin{equation}
	    	r(z) := \frac{b(z)}{a(z)}. \label{reflec}
	    \end{equation}
The scattering coefficients  and the reflection coefficients have  the following  properties \cite{zx1}.

\begin{lemma} \label{lsym2}
  Let $q_0(x)-\tanh{(x)}\in L^{1,2}(\mathbb{R})$ and $q_0'(x) \in W^{1,1}(\mathbb{R})$. Then
\begin{itemize}

        \item The scattering coefficients can be expressed in terms of the Jost functions as
        \begin{equation}\label{azbz}
            a(z)=\frac{{\rm det}(\Phi_{+,1}, \Phi_{-,2})}{1-z^{-2}}, \quad b(z)=\frac{{\rm det}(\Phi_{-,1}, \Phi_{+,1})}{1-z^{-2}},
        \end{equation}
        where $\Phi_{\pm,j}(z),\, j=1,2,$ are the $j^{\text{th}}$ column of $\Phi_{\pm}(z)$.
    \item   $a(z)$ can be analytically extended to $ \mathbb{C}^+$.  Moreover, the
    zeros of $a(z)$ in  $\mathbb{C}^+ $ are simple, finite, and located on the unit circle. $b(z)$ and $r(z)$ are defined only for $z\in \mathbb{R}\setminus\{0, \pm 1\}$.

        \item For each $z\in \mathbb{R}\setminus\{0, \pm 1\}$, we have
        \begin{equation}\label{relazbz}
            {\rm det}S(z)=|a(z)|^2-|b(z)|^2=1, \quad |r(z)|^2=1-|a(z)|^{-2}<1.
        \end{equation}
        \item  $a(z)$, $b(z)$, and $r(z)$ satisfy the symmetries
        \begin{align}
            &a(z)=\overline{a(-\bar{z})}=-\overline{a(\bar{z}^{-1})},\\
            &b(z)=\overline{b(-\bar{z})}=\overline{b(\bar{z}^{-1})},\\
            &r(z)=\overline{r(-\bar{z})}=-\overline{r(\bar{z}^{-1})}.
        \end{align}
        \item The scattering data has the following asymptotics
        \begin{align}
        &\lim_{z\rightarrow\infty}(a(z)-1)z=i\int_{\mathbb{R}}(q^2-1)\mathrm{d}x, \ z\in\overline{\mathbb{C}^+}, \label{asymptoticsfora1}\\
        &\lim_{z\rightarrow0}(a(z)+1)z^{-1}=i\int_{\mathbb{R}}(q^2-1)\mathrm{d}x, \ z\in\overline{\mathbb{C}^+},\label{asymptoticsfora2}\\
        &|b(z)|=\mathcal{O}(|z|^{-2}),\hspace{0.5cm} \text{as } |z|\rightarrow\infty,\ \ z\in\mathbb{R}, \label{asymptoticsforb1}\\
        &|b(z)|=\mathcal{O}(|z|^{2}),\hspace{0.5cm} \text{as } |z|\rightarrow0, \ z\in\mathbb{R},\label{asymptoticsforb2}\\
        &r(z)\sim z^{-2},\  |z|\rightarrow\infty,  \ \ r(z)\sim z^2, \ |z|\rightarrow 0.
        \end{align}
\end{itemize}
\end{lemma}

In the generic case, although $a(z)$ and $b(z)$ have singularities at points $\pm 1$, the reflection coefficient $r(z)$ remains bounded
at $z=\pm 1$ with $|r(\pm 1)|=1$. Indeed, as $z\rightarrow \pm 1$,
\begin{align}
    &a(z)=\pm\frac{ s_{\pm}}{z\mp 1}+\mathcal{O}(1),\ \ b(z)=-\frac{i s_{\pm}}{z\mp 1}+\mathcal{O}(1), \nonumber
\end{align}
where $s_{\pm}=\frac{1}{2}{\rm det}\left(\Phi_{+,1}(\pm 1), \Phi_{-,2}(\pm 1)\right)$. Then,
\begin{equation}\label{rzlim}
    \lim_{z\rightarrow\pm 1}r(z)=\mp i.
\end{equation}
While in the non-generic case, $a(z)$ and $b(z)$ are continuous at $z=\pm 1$ with $|r(\pm 1)|<1$.

It can be shown that the following lemma holds \cite{zx1,zx2}.
\begin{lemma}
	Given $q_0(x)-\tanh{(x)}\in L^{1,2}(\mathbb{R})$ and $q'_0(x) \in W^{1,1}(\mathbb{R})$, then $r(z)  \in H^{1}(\mathbb{R})$.
\end{lemma}

We now turn our attention to the discrete spectrum. Let  $\nu_1, \nu_2,\dots, \nu_N$ denote the $N$ zeros of $a(z)$ lying on $\mathbb{C}^{+}\cap \{z:|z|=1, \im z>0, \re z>0\}$.
The symmetries of $a(z)$ imply that the discrete spectrum is  collected as
\begin{equation}
    \mathcal{Z}=\{\nu_n, \bar{\nu}_n -\bar{\nu}_n, -\nu_n\}_{n=1}^{N},
\end{equation}
where $\nu_n$ satisfies that $|\nu_n|=1$, $\re \nu_n>0$, $\im \nu_n>0$.
Moreover, it is convenient to define that
\begin{equation}
    \eta_n=\left\{
        \begin{aligned}
        &\nu_n,\quad &n=1, \dots, N, \\
        &-\bar{\nu}_{n-N},\quad &n=N+1 ,\dots, 2N,
        \end{aligned}
        \right.
\end{equation}
from which we express the set $\mathcal{Z}$ in terms of
\begin{equation}
\mathcal{Z}=\{\eta_n, \bar{\eta}_n\}_{n=1}^{2N}.
\end{equation}
The distribution of $\mathcal{Z}$ on the $z$-plane is shown in Figure \ref{spectrumsdis}.
\begin{figure}[H]
	\centering
	\begin{tikzpicture}[node distance=2cm]
	\draw[-latex](-3.5,0)--(3.5,0)node[right]{$\re z$};
	\draw[-latex](0,-3)--(0,3)node[above]{$\im z$};
    \node   at (1.1,2.5) {\footnotesize $\mathbb{C}^+$};
    \node   at (1.1,-2.5) {\footnotesize $\mathbb{C}^-$};
	\draw[dotted] (2,0) arc (0:360:2);
	\coordinate (A) at (2,2.985);
    \coordinate (B) at (2,-2.985);
    \coordinate (C) at (-0.616996232,0.9120505887);
    \coordinate (D) at (-0.616996232,-0.9120505887);
    \coordinate (E) at (0.616996232,0.9120505887);
    \coordinate (F) at (0.616996232,-0.9120505887);
    \coordinate (G) at (-2,2.985);
    \coordinate (H) at (-2,-2.985);
	\coordinate (J) at (1.7570508075688774,0.956);
	\coordinate (K) at (1.7570508075688774,-0.956);
	\coordinate (L) at (-1.7570508075688774,0.956);
	\coordinate (M) at (-1.7570508075688774,-0.956);
    \coordinate (N) at (0, 2);
    \coordinate (O) at (0, -2);
	\fill[red] (J) circle (1pt) node[right] {$\nu_n$};
	\fill[red] (K) circle (1pt) node[right] {$\bar{\nu}_n$};
	\fill[red] (L) circle (1pt) node[left] {$-\bar{\nu}_n$};
	\fill[red] (M) circle (1pt) node[left] {$-\nu_n$};
 \fill[red] (N) circle (1pt);
\node  [red, above]  at (0.15,2) {$i$};
\fill[red] (O) circle (1pt);
\node  [red, below]  at (-0.25,-2) {$-i$};
	\end{tikzpicture}
	\caption{\footnotesize The distribution of the discrete spectrums on the unit circle $\{z:|z|=1\}$  in  the $z$-plane.
     }
	\label{spectrumsdis}
\end{figure}
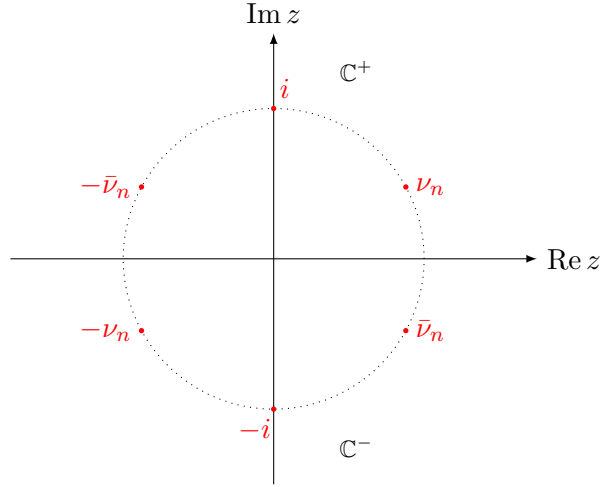

 Moreover,  we have the trace formula of $a(z)$:
    \begin{equation}\label{traceformula}
        a(z)=\prod_{n=1}^{2N}\left(\frac{z-\eta_n}{z-\bar{\eta}_n}\right){\rm exp}\left(-\frac{1}{2\pi i}\int_{\mathbb{R}}\frac{{\rm log}(1-|r(\zeta)|^2)}{\zeta-z}\mathrm{d}\zeta\right), \quad z\in\mathbb{C}^+.
    \end{equation}
At any zero $z=\eta_n \in \overline{\mathbb{C}^+}$ of $a(z)$, it follows from \eqref{azbz} that the pair $\Phi_{+,1}(\eta_n)$ and $\Phi_{-,2}(\eta_n)$ are linearly related. Moreover, the symmetry \eqref{symmu} implies that  $\Phi_{+,2}(\eta_n)$ and $\Phi_{-,1}(\eta_n)$ are also linearly related. Thus, there exsits a constant $\gamma_n \in \mathbb{C}$  such that
\begin{equation}
\Phi_{+,1}(\eta_n) = \gamma_n \Phi_{-,2}(\eta_n),\quad \Phi_{+,2}(\bar{\eta}_n)= \bar{\gamma}_n \Phi_{-,1}(\bar{\eta}_n).
\end{equation}
These constants $\gamma_n$ are referred to as the connection coefficients associated with the discrete spectral values $\eta_n$.

\subsection{Time evolution of the scattering data }

In order to solve the Cauchy  problem (\ref{dmkdv})-(\ref{bdries})  for  the defocusing mKdV equation,
 we need to determine the time dependence of
the scattering data.  For $q(x,t)$, the solution to  (\ref{dmkdv}), and the time-dependent
Jost function $\Phi(z;x,t)$,  the compatibility condition for Lax pair (\ref{lax pair}) can be written in the form
\begin{align}
\frac{d}{dt} ( \partial_x-X)=[T,  \partial_x-X],\label{oepe}
\end{align}
which is applied to  the first equation of the Lax pair (\ref{lax pair}),  we obtain
$$ ( \partial_x-X )( \Phi_t(z;x,t)- T\Phi(z;x,t))=0,$$
and hence
\begin{align}
 (\partial_t  - T ) \Phi_\pm (z;x,t)=\Phi_\pm (z;x,t) C_\pm(z,t),  \label{jjd}
 \end{align}
where $C_\pm(z,t)$ is a matrix function  to be determined.  By using the transformation  (\ref{muPhirelation}),   we write  (\ref{jjd}) in the form
\begin{align}
 (\partial_t - T)\mu_{\pm}(z;x,t)  = \mu_{\pm}(z;x,t) e^{i\lambda(z)x\widehat\sigma_3} C_\pm(z,t). \label{jheh}
 \end{align}
Then, using the asymptotics
\begin{align}
  &  \mu_{\pm}(z;x,t) \to E_\pm, \ \ \partial_t \mu_{\pm}(z;x,t) \to 0,  \ \  x\to \pm \infty, \nonumber\\
    & T\to  (4k^2+2) (ik\sigma_3\pm \sigma_1), \ \  x\to \pm \infty, \nonumber
\end{align}
it follows from (\ref{jheh}) that
$$C_\pm(z,t)= - (4k^2+2) i \lambda \sigma_3. $$

Applying  $\partial_t-T$ to the scattering relation $\Phi_+ (z;x,t)=\Phi_- (z;x,t) S(z, t)  $, we get
$$ \partial_t S=2i\lambda(z)(4k^2(z)+2) [S,\sigma_3], $$
which yields
\begin{align}
&a(z,t)=a(z,0), \ b(z,t)=b(z,0)e^{-2i\lambda(z)(4k^2(z)+2)t}, \nonumber\\
&r(z,t)=r(z,0)e^{-2i\lambda(z)(4k^2(z)+2)t},\
\gamma_n(t)=\gamma_n(0) e^{-2i\lambda(z)(4k^2(z)+2)t}. \nonumber
\end{align}

\section{Inverse Scattering and  the RH Problem  }\label{rh}

\subsection{A basic RH problem}
For $z\in \mathbb{C} \setminus \mathbb{R}$ and the Jost functions $\mu_{\pm,j}(z;x,t),\, j=1,2$,
we define a sectionally meromorphic matrix as follows:
\begin{equation}
    M(z)=M(z;x,t):=\left\{
        \begin{aligned}
        \left(\frac{\mu_{+,1}(z;x,t)}{a(z)}, \mu_{-,2}(z;x,t)\right), \quad z\in \mathbb{C}^{+}, \\
        \left(\mu_{-,1}(z;x,t), \frac{\mu_{+,2}(z;x,t)}{\overline{a(\bar z)}}\right), \quad z\in \mathbb{C}^{-},
        \end{aligned}
        \right.
\end{equation}
which solves the following RH problem.

\begin{prob} \label{RHP0}
  Find  a matrix-valued function $M(z)=M(z;x,t)$ which satisfies
\begin{itemize}
\item  Analyticity: $M(z)$  is analytic in $\mathbb{C}\setminus (\mathbb{R}\cup\mathcal{Z})$ and has simple poles at the points in $\mathcal{Z}$.

\item  Jump condition: $M(z)$ satisfies the jump condition
	$$M_+(z)=M_-(z)V(z), \; z \in \mathbb{R},$$
      	where
      	\begin{equation}\label{V0}
      		V(z)=\left(\begin{array}{cc}
      			1-|r(z)|^2 & -\overline{r(z)}e^{2it \theta(z)}\\
      		r(z)	e^{-2it \theta(z)} & 1
      		\end{array}\right),
      	\end{equation}
      with
      $$ \theta(z) =\lambda(z)\left(\frac{x}{t}+4k^2(z)+2\right).$$


\item  Asymptotic behaviors:
\begin{align*}
&M(z)=I+\mathcal{O}(z^{-1}), \quad  z\rightarrow \infty,\\
&zM(z)=\sigma_2+\mathcal{O}(z), \quad  z\rightarrow 0.
\end{align*}

\item  Residue conditions:
\begin{align}
&\underset{z=\eta_n}{\rm Res}M(z)=\lim_{z\rightarrow\eta_{n}}M(z)\left(\begin{array}{cc} 0 & 0 \\ c_n e^{-2it\theta(\eta_{n})} & 0 \end{array}\right),\label{fresm1}\\
&\underset{z=\bar{\eta}_n}{\rm Res}M(z)=\lim_{z\rightarrow\bar{\eta}_{n}}M(z)\left(\begin{array}{cc} 0 & \bar{c}_ne^{2it\theta(\bar{\eta}_{n}}) \\ 0 & 0 \end{array}\right),\label{fresm2}
\end{align}
where
\begin{align}
&c_n = \frac{\gamma_n(0)}{a'(\eta_n)} = \frac{2\eta_n}{\int_{\mathbb{R}}|\Phi_{-,2}(\eta_n;x,0)|^2 \mathrm{d}x}=\eta_n|c_n|. \label{conect}
 \end{align}
\end{itemize}
\end{prob}
The potential $q(x,t)$ is given by the reconstruction formula
\begin{equation}
q(x,t)=\lim_{z\rightarrow\infty}i(zM (z ))_{21},\label{sol}
\end{equation}
where the subscript $21$ denotes the element in the  $2$-th row and $1$-th column of the matrix $M(z)$.
 Due to the symmetries contained in Lemma \ref{lsym1} and \ref{lsym2}, and the uniqueness of the solution of RH problem \ref{RHP0}, it follows that
 \begin{equation}\label{msym}
 M(z)=\sigma_1 M^*(z)\sigma_1=\overline{M(-\bar{z})}=\mp z^{-1}M(z^{-1})\sigma_2.
 \end{equation}

\subsection{Saddle points and the signature table}
   Two well-known factorizations of the jump matrix $V(z)$  in \eqref{V0} are as follows:
       \begin{equation}\label{v}
         V(z)=        \begin{cases}
           \left(\begin{array}{cc} 1&-\overline{r(z)}e^{2it\theta(z) }\\ 0&1\end{array}  \right)
           \left(\begin{array}{cc} 1&0\\   r(z)e^{-2it\theta(z) }&1\end{array}  \right),\\
           \left(\begin{array}{cc} 1 & 0 \\ \frac{r(z)}{1-|r(z)|^2}e^{-2it\theta(z) } & 1\end{array}  \right)
           \left(1-|r(z)|^2 \right)^{\sigma_3} \left(\begin{array}{cc} 1 &
             \frac{-\overline{r(z)}}{1-|r(z)|^2}e^{2it\theta(z) } \\ 0 & 1\end{array}  \right).
         \end{cases}
       \end{equation}
  The long-time asymptotics of RH problem \ref{RHP0} is affected by  the exponential
  function  $e^{\pm 2it\theta(z)}$ in the jump matrix $V(z)$ and the residue condition.
  Let $\xi:=\frac{x}{t}$ and $z=u+iv$, then direct calculation shows that
          \begin{equation}
         \re\left(2i\theta(z)\right) =-v \left((3u^2-v^2)\left(1+(u^2+v^2)^{-3}\right)+(\xi+3)\left(1+(u^2+v^2)^{-1}\right)\right).\label{theta01}
          \end{equation}
 The signature table   of $\re( 2 i\theta(z) )$ is   presented   in Figure \ref{proptheta}.  The sign of $\re (2i\theta(z))$  plays a crucial role in determining the growth/decay regions of    the  exponential function  $e^{\pm 2it\theta(z)}$. This observation motivates us to
 open the jump contour $\mathbb{R}$ using two different  factorizations  of  the jump matrix $V(z)$.

 \begin{figure}[htbp]
        \centering
              \subfigure[$ \xi<-6$]{\label{figurea}
          \begin{minipage}[t]{0.32\linewidth}
            \centering
            \includegraphics[width=1.5in]{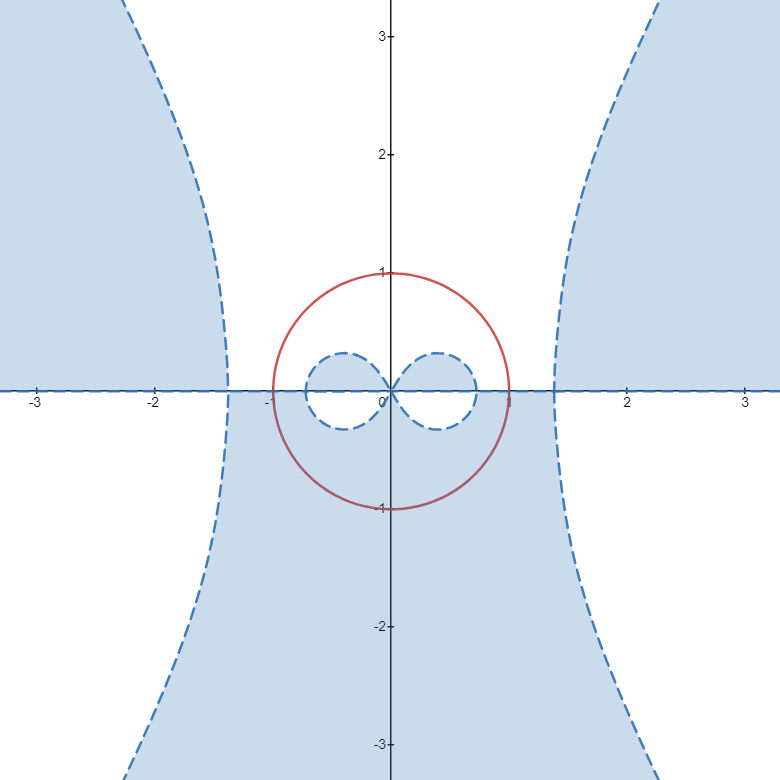}
          \end{minipage}
        }%
        \subfigure[$\xi=-6$]{\label{figureb}
          \begin{minipage}[t]{0.32\linewidth}
            \centering
            \includegraphics[width=1.5in]{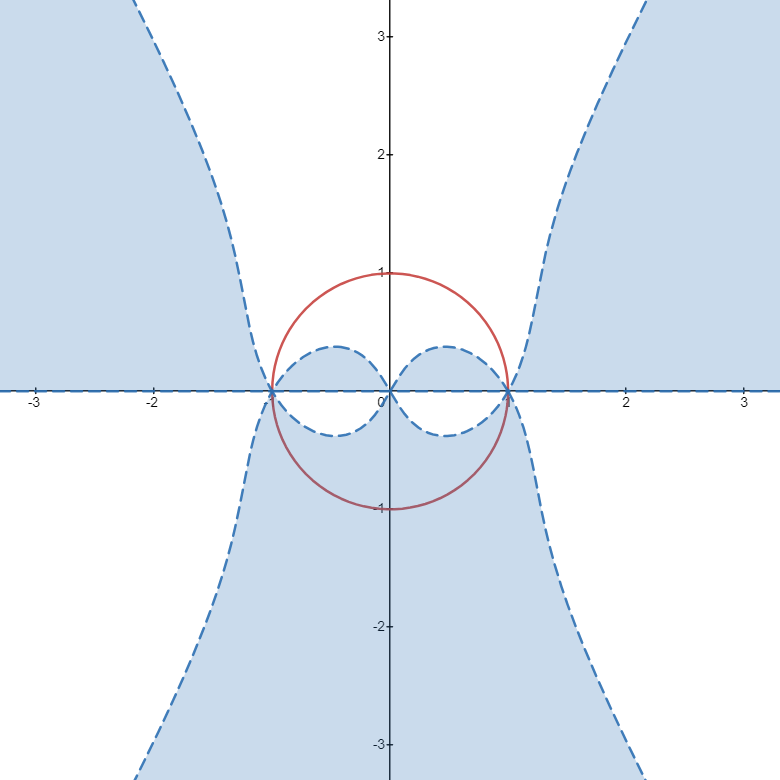}
          \end{minipage}%
        }%
         \subfigure[$-6<\xi<-2$]{\label{figurec}
        \begin{minipage}[t]{0.32\linewidth}
          \centering
          \includegraphics[width=1.5in]{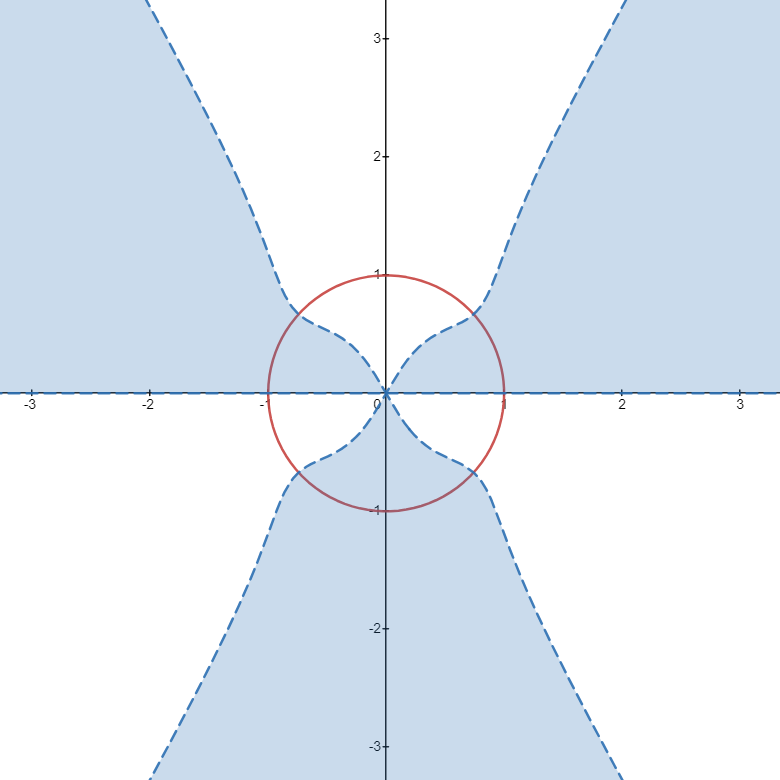}
        \end{minipage}
      }%

       \subfigure[$\xi=-2$]{\label{figured}
        \begin{minipage}[t]{0.32\linewidth}
          \centering
          \includegraphics[width=1.5in]{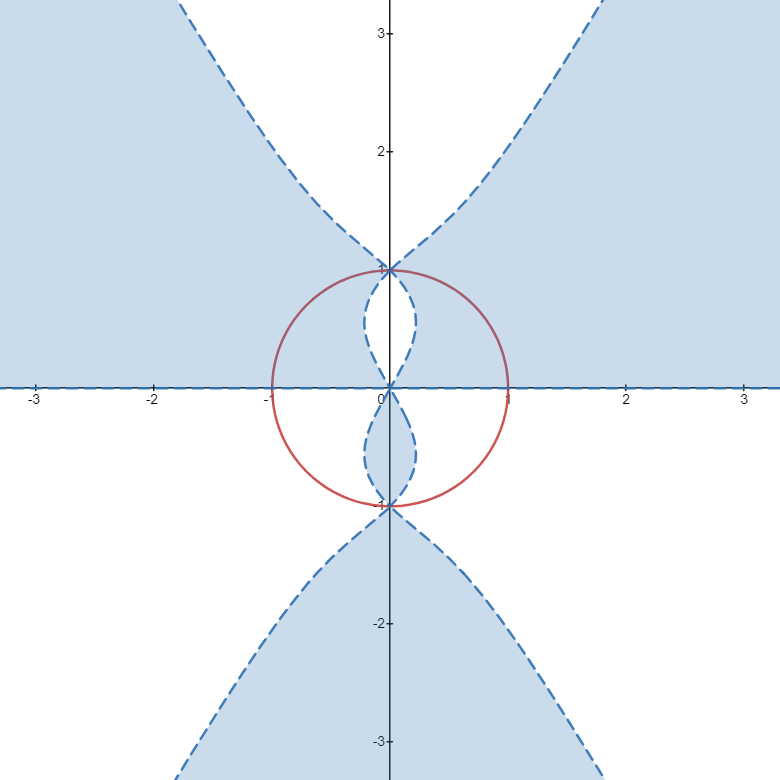}
        \end{minipage}
      }%
      \subfigure[$-2<\xi<6$]{\label{figuree}
    \begin{minipage}[t]{0.32\linewidth}
      \centering
      \includegraphics[width=1.5in]{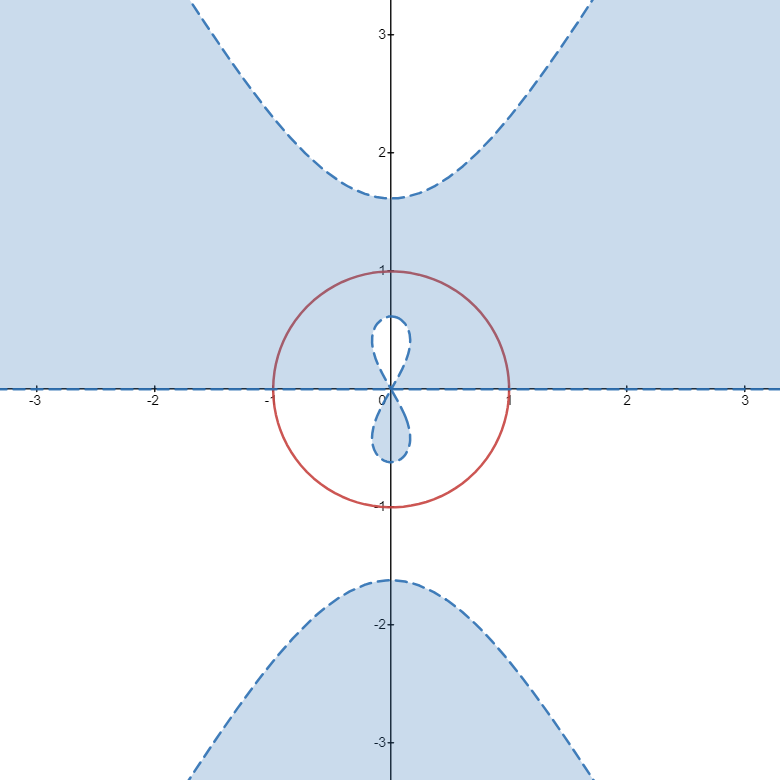}
    \end{minipage}
  }%
  \subfigure[$\xi\ge6$]{\label{figuref}
    \begin{minipage}[t]{0.32\linewidth}
      \centering
      \includegraphics[width=1.5in]{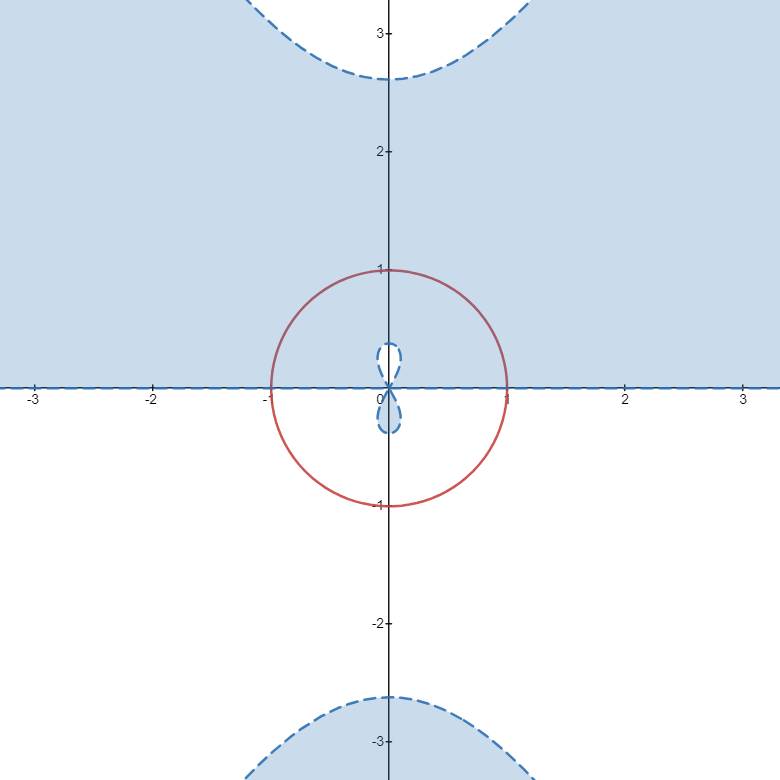}
    \end{minipage}
  }%
        \centering
        \caption{\footnotesize The distribution of saddle  points and the signature table of $\re (2i\theta(z))$, where  $\re (2i\theta(z))<0$   in blue regions and  $\re (2i\theta(z))>0$   in white regions.  Figure (a):  There are  four saddle points on $\mathbb{R}$ for the solitonless  region $\xi<-6$; Figure (c):  There are  four saddle points on the circle $|z|=1$   for the solitonic  region $-6<\xi\le-2$;  Figure (e) and (f): There are  four saddle points on $i\mathbb{R}$ for the solitonless region $\xi>-2$; Figures (b) and (d) are
        two critical cases.
       }
        \label{proptheta}
      \end{figure}

The saddle points (or stationary phase points) satisfy the following equation
  \begin{equation}
      \theta'(z)=\frac{(1+z^2)(3z^4+\xi z^2 +3)}{2z^4}=0,\label{2theta}
  \end{equation}
which  has the solutions
    \begin{equation}
  z^2 = -1 \ \ \text{or} \ \ z^2= \eta_+ \ \ \text{or}  \ \ z^2= \eta_-,\label{123}
  \end{equation}
  where
  $$\eta_\pm :=\frac{-\xi \pm \sqrt{\xi^2-36}}{6}, \ \ |\xi|> 6; \ \ \eta_\pm :=\frac{-\xi \pm i\sqrt{36-\xi^2}}{6}, \ \ |\xi|<6.$$
 Therefore,  we have two fixed saddle points $\pm i$ and
 other four saddle points which vary  with the value of $\xi$, distributed as follows:

  If  $\xi<-6$,   then  $\eta_\pm >0$  and four saddle points appear on the real axis
  \begin{equation}\label{eq219}
  z_1 =  \sqrt{\eta_+  }, \ z_2 =  \sqrt{\eta_-},\ z_3 = -\sqrt{\eta_- }, \ z_4 = - \sqrt{\eta_+}
  \end{equation}
    with $ z_4 <-1<z_3<0<z_2 <1<z_1$ and $z_1 z_2 =  z_3 z_4=1$.  See Figure \ref{figurea}.

 If  $\xi>6$,   then $\eta_\pm <0$ and  four saddle  points appear on the imaginary axis
    \begin{equation}\label{eq220}
  z_1 = i \sqrt{-\eta_-  }, \ z_2 = i \sqrt{-\eta_+  },\ z_3 = -i \sqrt{-\eta_+ }, \ z_4 = -i \sqrt{-\eta_- }
  \end{equation}
  with $ z_4 <-i< z_3<0<z_2 <i< z_1$ and $z_1 z_2 =  z_3 z_4=-1$.  See Figure \ref{figuref}.

   If  $-6<\xi<6$,    then   $   |\eta_\pm |=1$
    and   four saddle  points appear  on the  unit circle $|z|=1$
   $$ z_1= e^{ i\arg \eta_+/2},  \ \ z_2=  - e^{  i\arg \eta_+/2 }, \ \ z_3= e^{ i\arg \eta_-/2},  \ \ z_4=  - e^{  i\arg \eta_-/2 }.$$
See  Figure \ref{figurec} and \ref{figuree}.


 Denote  the critical line $\mathcal{L}:= \{z: \rm Re (2i\theta(z))=0\}$ and the unit circle $\mathcal{C}:= \{z: |z|=1\}$.
  By considering  the cross  points between   $\mathcal{L}$ and   $\mathcal{C}$,  (\ref{theta01})   simplifies to
  \begin{equation}\label{solv}
    2(u^2-v^2) +\xi+4  =0.
  \end{equation}
 From this equation,  we find that    the critical points are
   $z=\pm 1 $ on the real axis   when   $\xi=-6$  and   $z=\pm i $  on    the imaginary axis  when  $\xi=-2$, respectively.   Based on the interaction
  between   $\mathcal{L}$ and   $\mathcal{C}$,  we can classify the  asymptotic regions  as follows.

\begin{itemize}

\item {\bf Solitonless region:}   For the case $\xi<-6$  or   $	\xi>-2$,  there is no interaction between $\mathcal{L}$ and $\mathcal{C}$.
Moreover, $\mathcal{L}$ remains far way from $\mathcal{C}$, as depicted in
Figure \ref{figurea} and \ref{figuree}.  This case corresponds to two distinct solitonless regions, which have been discussed in \cite{zx1}.

\item {\bf Solitonic  region:}     For the case $\xi=-2$,  the  critical  points
  $z=\pm i$  do not appear on the contour $ \mathbb{R}$,   and   in fact are    just special  poles
when $\nu_n=-\bar \nu_n$.    Therefore,   for the case $-6<\xi\leq -2$,
 $\mathcal{L}$ interacts with $\mathcal{C}$, as shown in  Figure \ref{figurec} and \ref{figured}.
This  is a solitonic region, in which the soliton resolution and the stability  of $N$-solitons were  investigated in  \cite{zx2}.

\item {\bf Transition region:}
For the case $\xi = -6$,   the critical points are $z=\pm  1$ which arise from the pairwise coalescence of four saddle  points $z_j,\, j=1,2,3,4$,  as illustrated in  Figure  \ref{figureb}.
Moreover, for the generic case, $|r(\pm 1)|=1$,  it turns out that
the norm $(1-|r(\pm 1)|^2)^{-1}$ blows up as $z\to \pm 1$.  This indicates
the emergence of a new phenomenon  in the  transition region $\xi \approx -6$, which  will be the focus of our investigation in the present paper.

\end{itemize}

\subsection{A regular RH problem}\label{modi1}

We make two successive transformations to the basic  RH problem \ref{RHP0}  to obtain a regular RH problem without poles and  singularities.

{\bf Step 1: Removing poles.}
 Since the  poles $\eta_n$ and $\overline{\eta}_n \in \mathcal{Z}$ are finite,
 distributed on the unit circle, and  far away from the jump contour $\mathbb{R}$ and the critical line $\mathcal{L}$,
they  decay exponentially when we convert their residues  to jumps on small circles around the poles.  This allows us to modify the basic  RH problem \ref{RHP0} by removing these poles firstly.

To   open   the contour $\mathbb{R}$  by  the  second  matrix   decomposition in  (\ref{v}),
we  define  the following scalar function
      \begin{equation}\label{delta}
    \delta(z)={\rm exp}\left(-i\int_{\mathbb{R}}\frac{\nu(\zeta)}{\zeta-z}\, \mathrm{d}\zeta\right),
        \end{equation}
where   $\nu(\zeta) = -\frac{1}{2\pi} \log (1-|r(\zeta)|^2)$.
It then follows that the subsequent proposition holds.

\begin{proposition} [\cite{zx2}] \label{prodel}
            The function $\delta(z)$ defined by (\ref{delta}) possesses the following properties:
            \begin{itemize} \label{prop1}
                \item $ \delta(z)$ is  analytic  in $\mathbb{C} \setminus  \mathbb{R}$.
                \item $\delta(z)=\overline{\delta^{-1}(\bar{z})}=\overline{\delta(-\bar{z})}=\delta(z^{-1})^{-1}$.
                \item $\delta_{-}(z)=\delta_{+}(z)\left(1-|r(z)|^2\right), \ z\in \mathbb{R}$.
                \item  The asymptotic behavior as $z\to \infty$ is
                \begin{equation}
                   \delta(\infty):=\lim_{z\rightarrow\infty}\delta(z)=1.\label{Texpan}
                \end{equation}
                \item $\frac{a(z)}{\delta(z)}$ is holomorphic and its absolute value is bounded in $\mathbb{C}^+$. Moreover, $\frac{a(z)}{\delta(z)}$ extends as a continuous function  and its absolute value equals to $1$ for $z \in \mathbb{R}$.
%

            \end{itemize}
        \end{proposition}

Define
\begin{equation}\label{definerho}
 \rho <  \frac{1}{2} {\rm min}  \{   \operatorname*{min}\limits_{ \eta_n,\, \eta_j\in  \mathcal{Z} }  |\eta_n-\eta_j|,    \operatorname*{min}\limits_{\eta_n\in  \mathcal{Z}} |\im \eta_n|,  \operatorname*{min}\limits_{\eta_n\in \mathcal{Z}, \, z \in \mathcal{L} }| \eta_n-z|  \}.
 \end{equation}
  For $\eta_n, \,\bar\eta_n \in \mathcal{Z}$, we make small circles  $C_n$ and  $\bar{C}_n$ centered at  $\eta_n $ and  $\bar \eta_n$ respectively,  with a radius of $\rho$. The corresponding disks $D_n$ and  $\bar{D}_n$ lie inside the domain with  $\re (2i\theta(z))>0$ for $\im z>0$ and $\re (2i\theta(z))<0$ for $\im z<0$. The circles are oriented counterclockwise in $\mathbb{C}^+$ and clockwise in $\mathbb{C}^-$.
See Figure \ref{Djump65}.

  In order to interpolate the poles trading them for jumps on  $C_n$ and  $\bar{C}_n$,
we construct the interpolation function
            \begin{equation}
           G(z) = \begin{cases}
            \left(\begin{array}{cc} 1&0 \\ -\displaystyle { \frac{c_n e^{ -2it\theta(\eta_n) }}{z-\eta_n}}&1\end{array}  \right), \;   z \in D_n,\\
             \left(\begin{array}{cc} 1& - \displaystyle {\frac{\bar c_n e^{ 2it\theta(\bar \eta_n) }}{z-\bar \eta_n} }\\ 0&1\end{array}  \right), \;   z \in \bar{D}_n,   \\
             I, \;  \;  \; \text{elsewhere},
           \end{cases}
       \end{equation}
  where $\eta_n, \,\bar\eta_n \in \mathcal{Z}$.

 Define
\begin{equation*}
	\Sigma^{(1)}=\mathbb{R} \cup \left( \bigcup_{n=1}^{2N} \left(  C_n \cup \bar{C}_n \right) \right),
\end{equation*}
where the direction on $\mathbb{R}$ goes from left to right, as shown in Figure \ref{Djump65}.
For convenience, let
\begin{equation*}
\Gamma = (-\infty,z_4) \cup (z_3,0) \cup (0,z_2) \cup (z_1,\infty).
\end{equation*}

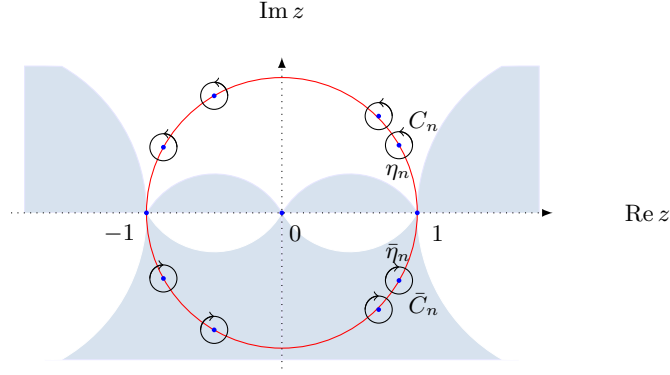
\begin{figure}
	\begin{center}
		\begin{tikzpicture}[scale=0.9]
			
			\draw[blue!10, fill=LightSteelBlue!50] (2,0) arc [start angle=30, end angle=150, radius =1.15];
			\draw[blue!10, fill=LightSteelBlue!50] (0,0) arc [start angle=30, end angle=150, radius =1.15];
			\draw[blue!10, fill=LightSteelBlue!50](-2,0) arc [start angle=0, end angle=60, radius =2.5]--(-3.5,2.17)--(-3.8,2.17)--(-3.8,0);
			\draw[blue!10, fill=LightSteelBlue!50] (2,0) arc [start angle=180, end angle=120, radius =2.5]--(3.5,2.17)--(3.8,2.17)--(3.8,0);
			\draw[blue!10, fill=LightSteelBlue!50] (-3.5,0)--(0,0)--(0,-2.17)--(-3.5,-2.17);
			\draw[blue!10, fill=LightSteelBlue!50] (3.5,0)--(0,0)--(0,-2.17)--(3.5,-2.17);
			\draw[white, fill=white](-2,0) arc [start angle=360, end angle=300, radius =2.5]--(-3.5,-2.17)--(-3.8,-2.17)--(-3.8,0);
			\draw[white, fill=white] (2,0) arc [start angle=180, end angle=240, radius =2.5]--(3.5,-2.17)--(3.8,-2.17)--(3.8,0);
			\draw[white, fill=white] (2,0) arc [start angle=330, end angle=210, radius =1.15];
			\draw[white, fill=white] (0,0) arc [start angle=330, end angle=210, radius =1.15];

			\draw [red,dotted,thick] (0, 0) circle [radius=2];
			\draw[   -latex,dotted ](-5, 0)--(5, 0);
			\draw[   -latex,dotted ](0, -2.3)--(0, 2.6);
			\node    at (5.4, 0)  {\footnotesize $\re z$};
			\node    at (0, 3)  {\footnotesize $\im z$};
			\node    at (1.72,  0.6)  {\footnotesize $\eta_n$};
           \node    at (2.1,  1.35)  {\footnotesize $C_n$};
			\node    at (1.72, -0.6)  {\footnotesize $\bar \eta_n$};
            \node    at (2.1,  -1.35)  {\footnotesize $\bar C_n$};
			\node    at (0.2, -0.3)  {\footnotesize $0$};
			\node    at (2.3, -0.3)  {\footnotesize $1$};
			\node    at (-2.4, -0.3)  {\footnotesize $-1$};
			\coordinate (A) at (1.73,  1);
			\coordinate (B) at (1.73,  -1);
			\coordinate (C) at (1.43,  1.43);
			\coordinate (D) at (1.43,  -1.43);
			\coordinate (E) at (-1.75,  0.97);
			\coordinate (F) at (-1.75,  -0.97);
			\coordinate (G) at (-1,  1.73);
			\coordinate (H) at (-1,  -1.73);
			\coordinate (K) at (0,  0);
			\coordinate (L) at (2,  0);
			\coordinate (M) at (-2,  0);
			\fill[blue] (A) circle (1pt);
			\fill[blue] (B) circle (1pt);
			\fill[blue] (C) circle (1pt);
			\fill[blue] (D) circle (1pt);
			\fill[blue] (E) circle (1pt);
			\fill[blue] (F) circle (1pt);
			\fill[blue] (G) circle (1pt);
			\fill[blue] (H) circle (1pt);
			\fill[blue] (K)  circle (1pt);
			\fill[blue] (L) circle (1pt);
			\fill[blue] (M) circle (1pt);
			\draw [] (A) circle [radius=0.2];
			\draw [  -> ]  (1.93,  1) to  [out=90,  in=0] (1.73,  1.2);	
			\draw [] (B) circle [radius=0.2];
			\draw [   -> ]  (1.53,  -1) to  [in=180,  out=90] (1.73,  -0.8);	
			\draw [] (C) circle [radius=0.2];
			\draw [   -> ]  (1.63,  1.43) to  [out=90,  in=0] (1.43,  1.63);	
			\draw [] (D) circle [radius=0.2];
			\draw [  -> ]  (1.23,  -1.43) to  [in=180,  out=90] (1.43,  -1.23);	
			\draw [] (E) circle [radius=0.2];
			\draw [  -> ]  (-1.55,  0.97) to  [out=90,  in=0]  (-1.75,  1.17);	
			\draw [] (F) circle [radius=0.2];
			\draw [  -> ]  (-1.95, -0.97) to  [in=180,  out=90]  (-1.75,  -0.77);	
			\draw [] (G) circle [radius=0.2];
			\draw [-> ]  (-0.8,  1.73) to  [out=90,  in=0]   (-1,  1.93);
			\draw [] (H) circle [radius=0.2];
			\draw [ -> ]  (-1.2,  -1.73) to  [in=180,  out=90]   (-1,  -1.53);
		\end{tikzpicture}
	\end{center}
	\caption{\footnotesize The jump contour  $\Sigma^{(1)}$  for  $M^{(1)}(z)$. In the blue regions, $\re (2i\theta(z))<0$, while in the white regions, $\re (2i\theta(z))>0$. }
	\label{Djump65}
\end{figure}

Denoting the factorization of  jump matrix  by
\begin{align}
\left(\begin{array}{cc}
		1 & 0\\
		\frac{ r(z) \delta_-(z)^{2} }{1-|r(z)|^2} e^{-2it\theta(z)} & 1
	\end{array}\right)\left(\begin{array}{cc}
	1 & -\frac{ \overline{r(z)}\delta_+(z)^{-2}}{1-|r(z)|^2} e^{2it\theta(z)}\\
	0 & 1
\end{array}\right):=B_-^{-1}B_+,  \label{opep2}
\end{align}
and  making the  transformation
       \begin{equation}
        M^{(1)}(z)= M(z) G(z)\delta(z)^{\sigma_3}, \label{trans1}
       \end{equation}
then $M^{(1)}(z)$ satisfies the symmetries of \eqref{msym} and  RH problem as follows.

 \begin{prob} \label{m1}
 Find  $M^{(1)}(z)=M^{(1)}(z;x,t)$ with properties
       \begin{itemize}
        \item   $M^{(1)}(z)$ is  analytic  in $ \mathbb{C}\setminus \Sigma^{(1)}$.
        \item  Jump condition:
        \begin{equation*}
            M^{(1)}_+(z)=M^{(1)}_-(z)V^{(1)}(z),
        \end{equation*}
        where
\begin{equation*}
	V^{(1)}(z)=\left\{\begin{array}{ll}
B_-^{-1}B_+,      & z\in \Gamma,\\[8pt]
\delta_-(z)^{-\sigma_3} V(z)\delta_+(z)^{\sigma_3},    & z\in \mathbb{R}\setminus \Gamma,\\[6pt]
		 \left(\begin{array}{cc} 1&0 \\ -\displaystyle { \frac{c_n e^{ -2it\theta(\eta_n) }\delta^2(z)}{z-\eta_n}}&1\end{array}  \right),  &z \in C_n,\, n=1,\cdots, 2N,\\[12pt]
		 \left(\begin{array}{cc} 1& \displaystyle {\frac{\bar c_n e^{ 2it\theta(\bar \eta_n) }\delta^{-2}(z)}{z-\bar \eta_n} }\\ 0&1\end{array}  \right), &z \in \bar{C}_n,\, n=1,\cdots, 2N.
	\end{array}\right.
\end{equation*}
        \item Asymptotic behaviors:
        \begin{align*}
                &M^{(1)}(z)=I+\mathcal{O}(z^{-1}),	\quad  z \to  \infty,\\
                &zM^{(1)}(z)=\sigma_2+\mathcal{O}(z), \quad z \to 0.
        \end{align*}

    \end{itemize}
\end{prob}

  Since   the jump matrices on the circles $C_n$ and $\bar C_n$ exponentially decay to the
  identity matrix as $t \to  \infty$,     RH problem \ref{m1} can be approximated by the following RH problem.
 \begin{prob} \label{m2}
 Find   $M^{(2)}(z)=M^{(2)}(z;x,t)$ with properties
       \begin{itemize}
        \item   $M^{(2)}(z)$ is  analytic  in $\mathbb{C}\setminus \mathbb{R}$.
        \item Jump condition:
        \begin{equation*}
            M^{(2)}_+(z)=M^{(2)}_-(z)V^{(2)}(z),
        \end{equation*}
        where
\begin{equation}
	V^{(2)}(z)=\left\{\begin{array}{ll}
B_-^{-1}B_+,    \; z\in \Gamma,\\[6pt]
\delta_-(z)^{-\sigma_3} V(z)\delta_+(z)^{\sigma_3},  \;  z \in \mathbb{R}\setminus \Gamma.
	\end{array}\right. \label{jumpv2}
\end{equation}

        \item Asymptotic behaviors:
        \begin{align*}
                &M^{(2)}(z)=I+\mathcal{O}(z^{-1}),	\quad  z \to  \infty,\\
                &zM^{(2)}(z)=\sigma_2+\mathcal{O}(z), \quad z \to 0.
        \end{align*}

\item    $M^{(2)}(z)$  admits the symmetries
$$M^{(2)}(z)=\sigma_1 M^{(2)*}(z)\sigma_1 =\overline{M^{(2)}(-\bar{z})}=\mp z^{-1}M^{(2)}(z^{-1})\sigma_2.$$
    \end{itemize}
\end{prob}
It can be shown that    $M^{(1)}(z)$  is  asymptotically equivalent to $M^{(2)}(z)$.
 \begin{proposition}
  \begin{equation}
  M^{(1)}(z) =  M^{(2)}(z)  \left(I+\mathcal{O}(e^{-ct}) \right),
  \end{equation}
where $c>0$ is a constant.
  \end{proposition}

{\bf Step 2: Removing singularities.}
In order to remove the   singularity  at  $z=0$, we  make a transformation
\begin{align}
M^{(2)}(z)=\left( I+ \frac{\sigma_2}{z} M^{(3)}(0)^{-1} \right ) M^{(3)}(z),\label{trans3}
\end{align}
then $M^{(2)}(z)$ satisfies the  RH problem  \ref{m2}  if $M^{(3)}(z)$ satisfies the following RH problem.

 \begin{prob} \label{ms3}
 Find  $M^{(3)}(z)=M^{(3)}(z;x,t)$ with properties
       \begin{itemize}
        \item  $M^{(3)}(z)$ is  analytic  in $\mathbb{C}\setminus \mathbb{R}$.
        \item  Jump condition: $  M^{(3)}_+(z)=M^{(3)}_-(z)V^{(2)}(z),$
        where
        $V^{(2)}(z)$ is given by   (\ref{jumpv2}).
        \item Asymptotic behavior:  $  M^{(3)}(z)=I+\mathcal{O}(z^{-1}),	\quad  z \to  \infty.$

\item  $M^{(3)}(z)$ satisfies the symmetries
$$M^{(3)}(z)=\sigma_1 M^{(3)*}(z)\sigma_1 =\overline{M^{(3)}(-\bar{z})}=\sigma_1 M^{(3)}(0)^{-1}M^{(3)}(z^{-1})\sigma_1.$$
    \end{itemize}
\end{prob}

\begin{proof}
The proof here is similar to   that of  RH problem 3.3 in \cite{wfp}. Thus, we omit it.
\end{proof}

\section{Long-time  Analysis in  the Transition Region} \label{sec3}
In this section,
 we consider the asymptotics    in the region   $ -C<(\xi+6) t^{2/3} < 0$ with $C>0$ which corresponds to Figure \ref{figurea}.
In this case,
the two saddle points $z_1$ and $z_2$ defined by (\ref{eq219}) are real and close to $z=1$ at least the speed of $t^{-1/3}$ as $t\to +\infty$.
Meanwhile, the other two saddle points $z_3$ and $z_4$ defined by (\ref{eq219})  are close to $z=-1$.

\subsection{A hybrid $\bar{\partial}$-RH problem} \label{modi2}
Fix a sufficiently small angle $\phi=\phi(\xi)$  such that  $\phi$ satisfies the following conditions:
\begin{itemize}
 \item $0<\phi < \arccos \frac{1}{\sqrt{5}-1} < \frac{\pi}{4}$;
 \item each $\Omega_j,\, j=0^\pm,1,2,3,4$ does not intersect with $\mathcal{L}$;
  \item each $\Omega_j,\, j=0^\pm,1,2,3,4$ does not intersect any small disks $D_n,\,\bar D_n,\, n=1,\cdots,2N$,
\end{itemize}
where $\Omega_j,\, j=0^\pm,1,2,3,4$ are defined by
\begin{align*}
&\Omega_{0^+}:= \{z\in\mathbb{C}:  0 \le \arg z \le \phi, \, |\re z| \le \frac{z_2}{2} \}, \ \Omega_{0^-} := \{z\in\mathbb{C}: -\bar{z} \in \Omega_{0^+} \},\\
& \Omega_1 := \{z\in\mathbb{C}: 0 \le \arg (z-z_1) \le \phi \}, \quad \Omega_4 := \{z\in\mathbb{C}: -\bar{z} \in \Omega_{1} \},\\
& \Omega_2 := \{z\in\mathbb{C}: \pi-\phi \le \arg (z-z_2) \le \pi, \, |\re (z-z_2)| \le \frac{z_2}{2}\}, \, \Omega_3 := \{z\in\mathbb{C}: -\bar{z} \in \Omega_{2} \},
\end{align*}
and   ${\Omega}^*_{j}$  denote  the conjugate regions of $\Omega_j$.
Moreover, to  open the jump contour $\Gamma$ by the $\bar{\partial}$ extension, we
define $\Sigma_{j},\, j=0^\pm,1,2,3,4$ as the boundaries of $\Omega_j$ and denote
\begin{equation*}
l \in \left(0, \frac{z_2}{2}\tan \phi\right), \quad \Sigma'_m= (-1)^{m+1}\frac{z_2}{2}+ e^{i\frac{\pi}{2}}l, \quad m=1,2.
\end{equation*}
 ${\Sigma}^*_{j},\,j=0^\pm,1,2,3,4$ and  ${\Sigma}'^*_m,\,m=1,2$  denote the  conjugate  contours above.
See   Figure \ref{signdbar}. Denote
  $$\Sigma = \bigcup_{j=0^\pm, 1,2,3,4}(\Sigma_j\cup \Sigma_j^*), \ \ \Sigma' = \bigcup_{m=1,2}(\Sigma'_m\cup \Sigma'^{*}_m),\ \ \Omega= \bigcup_{j=0^\pm, 1,2,3,4} (\Omega_j\cup {\Omega}^*_j).$$

\begin{figure}[htbp]
    \begin{center}
        \begin{tikzpicture}

        \draw[blue!10, fill=LightSteelBlue!50](-3.5,0) arc [start angle=0, end angle=60, radius =2.5]--(-4.8,2.21)--(-5.5,2.21)--(-5.5,0);
        \draw[blue!10, fill=LightSteelBlue!50] (3.5,0) arc [start angle=180, end angle=120, radius =2.5]--(4.8,2.21)--(5.5,2.21)--(5.5,0);
        \draw[blue!10, fill=LightSteelBlue!50] (2,0) arc [start angle=0, end angle=180, radius =1];
        \draw[blue!10, fill=LightSteelBlue!50] (0,0) arc [start angle=0, end angle=180, radius =1];
        \draw[blue!10, fill=LightSteelBlue!50](-3.5,0) arc [start angle=360, end angle=300, radius =2.5]--(-4.8,-2.21)--(4.8,-2.21)arc [start angle=240, end angle=180, radius =2.5]--(-3.5,0);
        \draw[white, fill=white] (2,0) arc [start angle=360, end angle=180, radius =1];
        \draw[white, fill=white] (0,0) arc [start angle=360, end angle=180, radius =1];

               \draw[-latex,dotted](-6,0)--(6,0)node[right]{ \textcolor{black}{$\re z$}};
               \draw[-latex,dotted](0,-2.5)--(0,2.5)node[right]{\textcolor{black}{$\im z$}};

               \coordinate (I) at (0,0);
               \coordinate (A) at (-3.5,0);
               \fill (A) circle (1pt) node[below] {\footnotesize $z_4$};
               \coordinate (b) at (-2,0);
               \fill (b) circle (1pt) node[below] {\footnotesize $z_3$};
               \coordinate (e) at (3.5,0);
               \fill (e) circle (1pt) node[below] {\footnotesize $z_1$};
               \coordinate (f) at (2,0);
               \fill (f) circle (1pt) node[below] {\footnotesize $z_2$};
               \draw[] (-3.5,0)--(-2,0);
               \draw[-latex] (-3.5,0)--(-2.9,0);
               \draw[] (3.5,0)--(2,0);
               \draw[-latex] (2,0)--(3.2,0);

               \coordinate (c) at (-2.65,0);
               \fill[] (c) circle (1pt) node[below] {\scriptsize$-1$};
               \coordinate (d) at (2.65,0);
               \fill[] (d) circle (1pt) node[below] {\scriptsize$1$};
               \draw[](0,0)--(1,0.8);
               \draw[-latex](0,0)--(0.5,0.4);
               \draw[](1,0.8)--(2,0);
               \draw[-latex](1,0.8)--(1.5,0.4);

               \draw[](3.5,0)--(5.5, 1.2);
               \draw[-latex](3.5,0)--(4.5,0.6);
               \draw[](0,0)--(1,-0.8);
               \draw[-latex](0,0)--(0.5,-0.4);
               \draw[](1,-0.8)--(2,0);
               \draw[-latex](1,-0.8)--(1.5,-0.4);

               \draw[](3.5,0)--(5.5, -1.2);
               \draw[-latex](3.5,0)--(4.5,-0.6);
               \draw[](0,0)--(-1,0.8);
               \draw[-latex](-1,0.8)--(-0.5,0.4);
               \draw[](-1,0.8)--(-2,0);
               \draw[-latex](-2,0)--(-1.5,0.4);

               \draw[](-3.5,0)--(-5.5, 1.2);
               \draw[  latex-](-4.5,0.6)--(-5.5, 1.2);
               \draw[](0,0)--(-1,-0.8);
               \draw[-latex](-1,-0.8)--(-0.5,-0.4);
               \draw[](-1,-0.8)--(-2,0);
               \draw[-latex](-2,0)--(-1.5,-0.4);

               \draw[](-3.5,0)--(-5.5, -1.2);
               \draw[ latex-](-4.5,-0.6)--(-5.5, -1.2);

               \draw[] (1,0)--(1,0.8);
                \draw[-latex] (1,0)--(1,0.4);
                 \draw[] (-1,0)--(-1,0.8);
                \draw[-latex] (-1,0)--(-1,0.4);
                \draw[] (1,0)--(1,-0.8);
                \draw[-latex] (1,0)--(1,-0.4);
                 \draw[] (-1,0)--(-1,-0.8);
                \draw[-latex] (-1,0)--(-1,-0.4);
               \node at (4.6,0.25)  {\scriptsize $\Omega_{1}$};
               \node at (4.6,0.9) {\scriptsize $\Sigma_{1}$};
                 \node at (4.6,-0.25)  {\scriptsize ${\Omega}^*_{1}$};
                  \node at (4.6,-0.9) {\scriptsize $\Sigma^*_{1}$};
                   \node at (-4.6,0.25)  {\scriptsize $\Omega_{4}$};
                    \node at (-4.6,0.9) {\scriptsize $\Sigma_{4}$};
                 \node at (-4.6,-0.25)  {\scriptsize ${\Omega}^*_{4}$};
                  \node at (-4.6,-0.9) {\scriptsize $\Sigma^*_{4}$};
               \node at (1.3,0.25)  {\scriptsize $\Omega_{2}$};
                \node at (1.5,0.68) {\scriptsize $\Sigma_{2}$};
                 \node at (1.3,-0.25)  {\scriptsize  ${\Omega}^*_{2}$};
                   \node at (1.5,-0.68) {\scriptsize $\Sigma^*_{2}$};
                  \node at (-1.3,0.25)  {\footnotesize $\Omega_{3}$};
                    \node at (-1.45,0.68) {\scriptsize $\Sigma_{3}$};
                 \node at (-1.3,-0.25)  {\scriptsize ${\Omega}^*_{3}$};
                  \node at (-1.45,-0.68) {\scriptsize $\Sigma^*_{3}$};
                  \node at (0.75,0.25)  {\scriptsize $\Omega_{0^+}$};
                   \node at (0.65,0.68) {\tiny $\Sigma_{0^+}$};
                 \node at (0.75,-0.25)  {\scriptsize ${\Omega}^*_{0^+}$};
                 \node at (0.65,-0.68) {\tiny $\Sigma^*_{0^+}$};
                                \node at (-0.65,0.25)  {\footnotesize $\Omega_{0^-}$};
                                \node at (-0.55,0.68) {\tiny $\Sigma_{0^-}$};
                 \node at (-0.65,-0.25)  {\scriptsize ${\Omega}^*_{0^-}$};
               \node at (-0.5,-0.68) {\tiny $\Sigma^*_{0^-}$};
               \end{tikzpicture}
          \caption{\footnotesize{Open the jump contour $\Gamma$. The regions where $\re \left(2i\theta(z) \right)<0$ are shown in blue, and those where $\re \left(2i\theta(z) \right)>0$ are in white.}}
      \label{signdbar}
    \end{center}
   \end{figure}
  To determine the decaying properties  of the oscillating factors   $ e^{\pm 2it\theta(z)}$,  we  especially  estimate
     $\re(2i\theta(z))$ in different regions.

\begin{proposition}\label{reprop1}  Let  $ -C< \left(\frac{x}{t}+6\right) t^{2/3}<0$. Denote $z=|z|e^{i\phi_0}$.
  Then the following estimates hold.
 \begin{itemize}
  \item $($corresponding to $z=0$$)$
    \begin{align}
    & \re(2i\theta(z))\leq   -c_0|\sin \phi_0| |\im z| ,\quad z\in \Omega_{0^+}\cup \Omega_{0^-},\label{estm1}\\
 & \re(2i\theta(z))\geq  c_0 |\sin \phi_0||\im z|,\quad z\in {\Omega}^*_{0^+}\cup {\Omega}^*_{0^-},
    \end{align}
    where $c_0=c_0(\phi_0,\xi)$ is a constant.
      \item $($corresponding to $z=z_j,\, j=1,4$$)$
    \begin{align}
       & \re(2i\theta(z))\leq \begin{cases}
       -c_j|\re z -z_j|^2 |\im z|,\quad z\in \Omega_{j} \cap \{ z:|z|\leq 2 \}, \label{omega1}\\
       -c_j|\im z|, \quad z \in \Omega_j \cap \{ z:|z|>2 \},
       \end{cases}\\
       & \re(2i\theta(z))\geq  \begin{cases}
      c_j|\re z -z_j|^2 |\im z|,\quad z\in  {\Omega}^*_{j} \cap \{ z:|z|\leq 2 \}, \\
    c_j|\im z|, \quad z \in {\Omega}^*_{j} \cap \{ z:|z|>2 \},
       \end{cases}
    \end{align}
 where $c_j=c_j(z_j, \phi_0,\xi)$ is a constant.
   \item $($corresponding to $z=z_j,\, j=2,3$$)$
                 \begin{align}
       &	\re(2i\theta(z))\leq -c_j|\re z -z_j|^2|\im z| ,\quad z\in \Omega_{j},	\\
      & 	\re(2i\theta(z))\geq c_j|\re z -z_j|^2|\im z| ,\quad z\in  {\Omega}^*_{j},
       \end{align}
       where $c_j=c_j(z_j, \phi_0,\xi)$ is a constant.
 \end{itemize}
  \end{proposition}
  \begin{proof}
   For the case corresponding to $z=0$,  we take   $\Omega_{0^+}$   as an  example  to prove the estimate (\ref{estm1}).
    Others can be proven in  a similar way.

 For $z\in \Omega_{0^+}$, denote the ray $z=|z|e^{i\phi_0}=u+iv$ where $0<\phi_0<\phi$ and $u>v>0$,  and the function $F(l)=l+l^{-1}$. Then,  (\ref{theta01}) becomes
    \begin{align}
        \re\left(2i\theta(z)\right)= - F(|z|) \sin\phi_0  \left( (1+2\cos 2\phi_0)F(|z|)^2 - 6 \cos 2\phi_0 + \xi \right). \label{eoue}
    \end{align}
Considering $(1+2\cos 2\phi_0)F(|z|)^2 - 6 \cos 2\phi_0 + \xi=0$, we have
\begin{align*}
   F(|z|)^2 = 3- \frac{3+\xi}{2\cos 2\phi_0 +1}:= \alpha >4.
\end{align*}
By $F(l)=\sqrt{\alpha}$, we have $l^2 -\sqrt{\alpha}l+1=0$. Solving the above equation, we find two roots $l_j$, $j=1,2$ with
\begin{equation*}
l_1 = \frac{\sqrt{\alpha}-\sqrt{\alpha-4}}{2} < l_2 = \frac{\sqrt{\alpha}+\sqrt{\alpha-4}}{2}.
\end{equation*}
Since $|z|\le \frac{z_2 \sec \phi_0}{2}<l_1$,
\begin{align*}
(1+2\cos 2\phi_0)F(|z|)^2 &- 6 \cos 2\phi_0 + \xi \ge \\
  &(1+2\cos 2\phi_0)F\left(\frac{z_2 \sec \phi_0}{2}\right)^2 - 6 \cos 2\phi_0 + \xi>0.
\end{align*}
Thus, there exists a constant $c_0 =c_0(\phi_0)$ such that
\begin{equation*}
 \re\left(2i\theta(z)\right) \le - c_0 F(|z|)\sin \phi_0.
\end{equation*}

For the cases corresponding to $z_j,\,j=1,4$, we take $\Omega_1$ as an example and others can be easily inferred.
Let $z=z_1+u+iv$. Then (\ref{theta01})  can be rewritten as
\begin{equation}
 \re\left(2i\theta(z)\right) = -v F(u,v),
\end{equation}
where
\begin{equation}\label{F1}
F(u,v)=(\xi+3) \left(1+|z|^{-2}\right) +\left(3(z_1+u)^2-v^2\right)\left(1+|z|^{-6}\right).
\end{equation}

For $z \in \Omega_1$ and $|z|\leq 2$, from (\ref{eq219}), we have
\begin{equation}\label{xixi1}
\xi = -3z_1^{-2} (1+z_1^4).
\end{equation}
Substituting  (\ref{xixi1}) into  (\ref{F1}) yields
\begin{equation}
F(u,v)= z_1^{-2} |z|^{-6} G(u,v),
\end{equation}
where
\begin{equation}
G(u,v)=3\left(z_1^2-1-z_1^4\right) \left(|z|^6+|z|^4\right) +\left(3z_1^2(z_1+u)^2-z_1^2v^2\right)\left(|z|^6+1\right).
\end{equation}
After simplification, we obtain $$G(u,v) \geq 16  z_1^2 u^2,$$
then $$F(u,v) \ge 16 u^2 |z|^{-6} \ \text{and}\ \re\left(2i\theta(z)\right) \le - 16 |z|^{-6}u^2v.$$ Since $|z|\le 2 $, $$\re\left(2i\theta(z)\right) \le -\frac{1}{4}u^2 v.$$
Moreover, the proof for the cases corresponding to $z_j,\, j=2,3$ can be given similarly.

Next we consider the estimate for $z \in \Omega_1$ and $|z| >2$.  In the transition region, as $\xi \to -6^-$,  (\ref{F1}) reduces to
\begin{equation}
F(u,v)=\left(1+|z|^{-6} \right) \left( -3 f(|z|) +3 (z_1 +u)^2 -v^2\right),
\end{equation}
where
\begin{equation}
f(x)=\frac{x^6+x^4}{x^6+1}.
\end{equation}
Since $f(x)$ has a maximum value $f_{max}=\frac{4}{3}$,
\begin{equation*}
F(u,v) \ge -4 +3(z_1+u)^2 -v^2.
\end{equation*}
Let $z=|z|e^{i\phi_0}$ with $0<\phi_0 <\phi$. By noting that $v=(z_1+u) \tan w$ where $0<w<\phi_0$, we obtain
\begin{equation*}
F(u,v) \ge -4 +2(z_1+u)^2 \ge -4+8 \cos^2 \phi_0 \ge -4+\frac{8}{(\sqrt{5}-1)^2}.
\end{equation*}
This completes the proof of the estimate (\ref{omega1}) in the domain $\Omega_1$.
  \end{proof}

Next we open  the contour $\Gamma$  via  continuous extensions of the jump matrix $V^{(2)}(z)$
by defining   appropriate functions.

 \begin{proposition} \label{prop3}

  Let $q_0(x)- \tanh(x)\in   H^{4,4}(\mathbb{R})$. Then it is possible to define   functions $R_{j}: \overline{\Omega}_j \to \mathbb{C},\ j=0^\pm,1,2,3,4$, continuous on  $\overline{\Omega}_j$, with continuous first partials on $\Omega_j$, and boundary values

       \begin{equation*}
       R_{j}(z)= \begin{cases}
           \frac{\overline{r(z)}\delta_+(z)^{-2}}{1-|r(z)|^2}, \quad  z \in \Gamma,\\
             \gamma(z_j),\quad z \in \Sigma_{j},
            \end{cases}
       \end{equation*}
       where
     \begin{align}\label{Rz}
 \gamma(z)=\begin{cases}
  \frac{\overline{S_{21}(z)}}{S_{11}(z)} \left( \frac{a(z)}{\delta_+(z)} \right)^2, \quad \text{\rm for the generic case},\\[5bp]
    \frac{\overline{r(z)}\delta_+(z)^{-2}}{1-|r(z)|^2}, \quad \text{\rm for the non-generic case},
    \end{cases}
\end{align}
with
\begin{align}
&S_{21}(z)= {\rm det}(\Phi_{-,1}(z), \Phi_{+,1}(z)), \label{S21}\\
&S_{11}(z) = {\rm det}(\Phi_{+,1}(z), \Phi_{-,2}(z)), \label{S11}
\end{align}
and $\gamma(0)=0$,
  such that  for $j=1,2$; a fixed constant $c=c(q_0)$; and a fixed cutoff function $\varphi \in C^\infty_0(\mathbb{R},[0,1])$ with small support near $1$; we have
 \begin{align}
  &|\bar{\partial}R_{j}(z)| \le
  c\left(  | r'\left( |z| \right)|+    | z-z_j |^{-1/2} + \varphi( |z|) \right),   \;z\in \Omega_j,\label{437}\\
 &|\bar{\partial}R_{j}(z)| \le c|z-1|,\;  z \in  \Omega_j \ \text{in a small fixed neighborhood of} \ 1;\label{estn1}
\end{align}
for $j=3,4$,  we have (\ref{437}) with $|z|$ replaced by $-|z|$ in the argument of $r'$ and $\varphi$, as well as (\ref{estn1});
for $j=0^\pm$, we have
 \begin{equation}
  |\bar{\partial}R_{0^\pm}(z)| \le
  c\left(  | r'\left( \pm |z| \right)|+    | z |^{-1/2}  \right),   \;z\in \Omega_0^\pm.
\end{equation}
The similar estimate holds for $|\bar{\partial}R^*_{j}(z)|$.

Setting $R: \Omega \to \mathbb{C}$ by $R(z)|_{z\in \Omega_j} =R_j(z)$ and $R(z)|_{z\in  {\Omega}^*_j} =R^*_j(z)$, the extension can preserve the symmetry $R(z)=-\overline{R( \bar{z}^{-1})}$.
        \end{proposition}

\begin{proof}
The proof follows a similar methodology to that outlined in \cite{CJ}.
For brevity, we omit it in this context.
\end{proof}

Define
 \begin{align}\label{R3}
    	R^{(3)}(z)=\begin{cases}
          \left( \begin{array}{cc}
    	1&R_j(z) e^{2it\theta(z)}\\0&1
    \end{array} \right),\quad z \in \Omega_j,\;  j=0^\pm,1,2,3,4, \\
    \left(\begin{array}{cc}
    	1&  0\\ {R}^*_j(z)e^{-2it\theta(z)}  &1      \end{array} \right), \quad z \in {\Omega}^*_{j},\;  j=0^\pm,1,2,3,4, \\
    I, \quad {\rm elsewhere},
    	\end{cases}
    \end{align}
and
$$\Sigma^{(4)}= \Sigma \cup \Sigma' \cup  [z_4,z_3]\cup [z_2, z_1].$$
See Figure \ref{jumpm4}.

  Then the new matrix-valued function
    \begin{equation}\label{trans4}
        M^{(4)}(z)=M^{(3)}(z)R^{(3)}(z)
    \end{equation}
 satisfies the   following hybrid  $\bar{\partial}$-RH problem.

\begin{prob2}\label{dbarrhp}
    Find    $M^{(4)}(z)=M^{(4)}(z;x,t)$  with properties
    \begin{itemize}
        \item  $M^{(4)}(z)$ is continuous in $\mathbb{C}\setminus  \Sigma^{(4)} $ and takes continuous boundary values $M^{(4)}_+(z)$ $($respectively  $M^{(4)}_-(z))$ on $\Sigma^{(4)}$ from the left $($respectively right$)$.
        \item $M^{(4)}(z)$ satisfies the  jump condition
        \begin{equation*}
            M^{(4)}_+(z)=M^{(4)}_-(z)V^{(4)}(z), \quad z \in\Sigma^{(4)},
        \end{equation*}
     \end{itemize}
        where
       \begin{align}\label{V3}
       	V^{(4)}(z)=\begin{cases}
           \left( \begin{array}{cc}
    	1& -\gamma(z_j) e^{2it\theta(z)}\\0&1
    \end{array} \right), \  z \in \Sigma_j, \ j=1,2,3,4,\vspace{2mm}\\
    \left(\begin{array}{cc}
    	1&0\\ \overline{\gamma(z_j)}e^{-2it\theta(z)}&1
    \end{array} \right), \  z \in {\Sigma}^*_j,\ j=1,2,3,4,\vspace{2mm}\\
     	\left(	\begin{array}{cc}
     	1& -\gamma(z_2)e^{2it\theta(z) }\\
     	0& 1
     \end{array}\right),\ z\in \Sigma'_{1}, \vspace{2mm}\\
            \left(		\begin{array}{cc}
       	1& 0\\
        \overline{\gamma(z_2)}e^{-2it\theta(z)} & 1
       \end{array}\right),\  z\in {\Sigma}^{*'}_{1},\vspace{2mm}\\
            	\left(	\begin{array}{cc}
     	1& \gamma(z_3)e^{2it\theta(z) }\\
     	0& 1
     \end{array}\right),\ z\in \Sigma'_{2}, \vspace{2mm}\\
            \left(		\begin{array}{cc}
       	1& 0\\
       - \overline{\gamma(z_3)}e^{-2it\theta(z)} & 1
       \end{array}\right),\  z\in {\Sigma}^{*'}_{2},\vspace{2mm}\\
       \delta_-(z)^{-\sigma_3} V(z) \delta_+(z)^{\sigma_3},\quad z\in \mathbb{R}\setminus \Gamma.
       	\end{cases}
       \end{align}
    \begin{itemize}
        \item $ M^{(4)}(z)=I+\mathcal{O}(z^{-1}),	\quad  z \to  \infty.$

        \item For $z\in \mathbb{C}$, we have
        \begin{equation*}
            \bar{\partial}M^{(4)}(z)= M^{(4)}(z) \bar{\partial}R^{(3)}(z),
        \end{equation*}
        where
        \begin{equation}\label{parR2}
            \bar{\partial}R^{(3)}(z)= \begin{cases}
                \left( \begin{array}{cc}
    	0& \bar{\partial}R_j(z)e^{2it\theta(z)}\\0&0
    \end{array} \right), \quad  z \in \Omega_j,\; j=0^\pm,1,2,3,4,\\
    \left(\begin{array}{cc}
    	0&  0\\\bar{\partial} {R}^*_j(z)e^{-2it\theta(z)}&0
    \end{array} \right), \quad z \in \Omega^*_j, \ j=0^\pm,1,2,3,4,\\
      0, \quad \text{elsewhere}.
            	\end{cases}
           \end{equation}

\end{itemize}
    \end{prob2}

Until now we have obtained the hybrid  $\bar{\partial}$-RH problem \ref{dbarrhp} for  $M^{(4)}(z)$  to analyze the long-time asymptotics of the original RH problem \ref{RHP0} for $M(z)$.
Next, we will construct the solution $M^{(4)}(z)$ as follows:
\begin{itemize}
 \item We first remove the  $\bar{\partial}$ component of the solution $M^{(4)}(z)$ and demonstrate the existence of a solution of the resulting pure RH problem. Furthermore, we calculate its asymptotic expansion.
  \item Conjugating off the solution of the first step,  a pure $\bar{\partial}$-problem can be obtained. Then, we establish the existence of a solution to this problem and  bound its magnitude.
 \end{itemize}


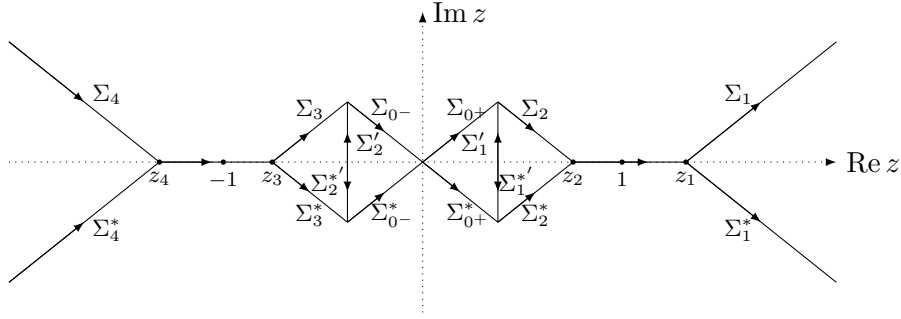
\begin{figure}[htbp]
    \begin{center}
        \begin{tikzpicture}
               \draw[-latex,dotted](-5.5,0)--(5.5,0)node[right]{ \textcolor{black}{$\re z$}};
               \draw[-latex,dotted](0,-2)--(0,2)node[right]{\textcolor{black}{$\im z$}};
               \coordinate (I) at (0,0);
               \coordinate (A) at (-3.5,0);
               \fill (A) circle (1pt) node[below] {\footnotesize $z_4$};
               \coordinate (b) at (-2,0);
               \fill (b) circle (1pt) node[below] {\footnotesize $z_3$};
               \coordinate (e) at (3.5,0);
               \fill (e) circle (1pt) node[below] {\footnotesize $z_1$};
               \coordinate (f) at (2,0);
               \fill (f) circle (1pt) node[below] {\footnotesize $z_2$};
               \draw[](-3.5,0)--(-2,0);
               \draw[-latex](-3.5,0)--(-2.8,0);
                \draw[](3.5,0)--(2,0);
               \draw[-latex](2,0)--(3,0);
               \coordinate (c) at (-2.65,0);
               \fill[] (c) circle (1pt) node[below] {\scriptsize$-1$};
               \coordinate (d) at (2.65,0);
               \fill[] (d) circle (1pt) node[below] {\scriptsize$1$};
               \draw[](0,0)--(1,0.8);
               \draw[-latex](0,0)--(0.5,0.4);
               \draw[](1,0.8)--(2,0);
               \draw[-latex](1,0.8)--(1.5,0.4);

               \draw[](3.5,0)--(5.5, 1.6);
               \draw[-latex](3.5,0)--(4.5,0.8);
               \draw[](0,0)--(1,-0.8);
               \draw[-latex](0,0)--(0.5,-0.4);
               \draw[](1,-0.8)--(2,0);
               \draw[-latex](1,-0.8)--(1.5,-0.4);

               \draw[](3.5,0)--(5.5, -1.6);
               \draw[-latex](3.5,0)--(4.5,-0.8);
               \draw[](0,0)--(-1,0.8);
               \draw[-latex](-1,0.8)--(-0.5,0.4);
               \draw[](-1,0.8)--(-2,0);
               \draw[-latex](-2,0)--(-1.5,0.4);

               \draw[](-3.5,0)--(-5.5, 1.6);
               \draw[latex-](-4.5,0.8)--(-5.5, 1.6);
               \draw[](0,0)--(-1,-0.8);
               \draw[-latex](-1,-0.8)--(-0.5,-0.4);
               \draw[](-1,-0.8)--(-2,0);
               \draw[-latex](-2,0)--(-1.5,-0.4);

               \draw[](-3.5,0)--(-5.5, -1.6);
               \draw[latex-](-4.5,-0.8)--(-5.5, -1.6);

               \draw[] (1,0)--(1,0.8);
                \draw[-latex] (1,0)--(1,0.4);
                 \draw[] (-1,0)--(-1,0.8);
                \draw[-latex] (-1,0)--(-1,0.4);
                \draw[] (1,0)--(1,-0.8);
                \draw[-latex] (1,0)--(1,-0.4);
                 \draw[] (-1,0)--(-1,-0.8);
                \draw[-latex] (-1,0)--(-1,-0.4);
               \node at (4.2,0.9)  {\footnotesize $\Sigma_{1}$};
                 \node at (4.2,-0.9)  {\footnotesize $\Sigma^*_{1}$};
                   \node at (-4.2,0.9)  {\footnotesize $\Sigma_{4}$};
                 \node at (-4.2,-0.9)  {\footnotesize $\Sigma^*_{4}$};
               \node at (1.5,0.7)  {\footnotesize $\Sigma_{2}$};
                 \node at (1.5,-0.7)  {\footnotesize $\Sigma^*_{2}$};
                \node at (-1.5,0.7)  {\footnotesize $\Sigma_{3}$};
                 \node at (-1.5,-0.7)  {\footnotesize $\Sigma^*_{3}$};
                  \node at (0.5,0.73)  {\footnotesize $\Sigma_{0^+}$};
                 \node at (0.6,-0.73)  {\footnotesize $\Sigma^*_{0^+}$};
                \node at (-0.4,0.73)  {\footnotesize $\Sigma_{0^-}$};
                 \node at (-0.4,-0.73)  {\footnotesize $\Sigma^*_{0^-}$};
                  \node at (0.7,0.25)  {\footnotesize $\Sigma'_{1}$};
                 \node at (1.25,-0.25)  {\footnotesize $\Sigma^{*'}_{1}$};
                 \node at (-0.7,0.25)  {\footnotesize $\Sigma'_{2}$};
                 \node at (-1.25,-0.25)  {\footnotesize $\Sigma^{*'}_{2}$};

               \end{tikzpicture}
            \caption{ \footnotesize{ The jump contour $\Sigma^{(4)}$ for  $M^{(4)}(z) $ and  $M^{rhp}(z) $.  }}
      \label{jumpm4}
        \end{center}
    \end{figure}

\subsection{Contribution from a pure RH problem} \label{modi3}
In this subsection, we first consider the  pure RH problem. Dropping the $\bar\partial$ component of $M^{(4)}(z)$,  $M^{rhp}(z)$  satisfies the  following pure RH problem.
\begin{prob}\label{mrhp}
    Find  $M^{rhp}(z)=M^{rhp}(z;x,t)$ which satisfies
	  \begin{itemize}
	  	\item  $M^{rhp}(z)$ is analytic in $\mathbb{C}\setminus \Sigma^{(4)}$. See   Figure \ref{jumpm4}.
	  	\item $M^{rhp}(z)$  satisfies the jump condition
\begin{equation*}
	  		M^{rhp}_+(z)=M^{rhp}_-(z)V^{(4)}(z),
	  	\end{equation*}
	  	where $V^{(4)}(z)$ is given by (\ref{V3}).
	  	\item   $M^{rhp}(z)$  has the same asymptotics with  $M^{(4)}(z)$.

	  \end{itemize}
\end{prob}

Based on the property of $V^{(4)}-I$, we analyze the local model $M^{loc}(z)$ of  $M^{rhp}(z)$ in the neighborhood of $z=\pm 1$.

\subsubsection{Local paramatrix} \label{localp}

Let  $t$  be  large enough so that $\sqrt{2C} (3t)^{-1/3+\tau}<\rho$ where $\tau$ is a constant with $0<\tau<\frac{1}{30}$ and $\rho$ has been defined in (\ref{definerho}).
For a fixed constant  $\varepsilon \leq \sqrt{2C}$, define  two open disks
  \begin{align}
&   \mathcal{U}_{r}  = \{z \in \mathbb{C}: |z-1|< (3t)^{-1/3+\tau}\varepsilon\},\ \ \mathcal{U}_{l }  = \{z \in \mathbb{C}: |z+1|< (3t)^{-1/3+\tau}\varepsilon\}.\nonumber
\end{align}
Denote the local jump contour
\begin{equation*}
  \Sigma^{loc}:= \Sigma^{(4)} \cap \left( \mathcal{U}_r  \cup \mathcal{U}_l \right),
\end{equation*}
as depicted in Figure \ref{scalingzk0}.
The local model  $M^{loc}(z)$  satisfies the following RH problem.

\begin{figure}[htbp]
    \begin{center}
        \begin{tikzpicture}
          \draw[-latex,dotted](-5.5,0)--(5.5,0)node[right]{ \textcolor{black}{$\re z$}};
               \draw[-latex,dotted](0,-2.5)--(0,2.5)node[right]{\textcolor{black}{$\im z$}};
               \coordinate (I) at (0,0);
               \coordinate (A) at (-3.2,0);
               \fill (A) circle (1pt) node[below] {\footnotesize $z_4$};
               \coordinate (b) at (-2.2,0);
               \fill (b) circle (1pt) node[below] {\footnotesize $z_3$};
               \coordinate (e) at (3.2,0);
               \fill (e) circle (1pt) node[below] {\footnotesize $z_1$};
               \coordinate (f) at (2.2,0);
               \fill (f) circle (1pt) node[below] {\footnotesize $z_2$};

               \coordinate (c) at (-2.7,0);
               \fill[] (c) circle (1pt) node[below] {\scriptsize$-1$};
               \coordinate (d) at (2.7,0);
               \fill[] (d) circle (1pt) node[below] {\scriptsize$1$};

               \draw[](-3.2,0)--(-2.2,0);
               \draw[-latex](-3.2,0)--(-2.8,0);
               \draw[](3.2,0)--(2.2,0);
               \draw[-latex](2.2,0)--(3,0);

               \draw[](1.6,0.48)--(2.2,0);
               \draw[-latex](1.8,0.32)--(2,0.16);
               \draw[](1.6,-0.48)--(2.2,0);
               \draw[-latex](1.6,-0.48)--(2,-0.16);
               \draw[](-1.6,0.48)--(-2.2,0);
               \draw[-latex](-2.2,0)--(-1.8,0.32);
               \draw[](-1.6,-0.48)--(-2.2,0);
               \draw[-latex](-2.2,0)--(-1.8,-0.32);

               \draw[](3.2,0)--(3.8,0.49);
               \draw[-latex](3.2,0)--(3.6,0.32);
               \draw[](3.2,0)--(3.8,-0.49);
               \draw[-latex](3.2,0)--(3.6,-0.32);

               \draw[](-3.2,0)--(-3.8,0.49);
               \draw[-latex](-3.8,0.49)--(-3.4,0.16);
               \draw[](-3.2,0)--(-3.8,-0.49);
               \draw[-latex](-3.8,-0.49)--(-3.4,-0.16);

               \node at (4,0.7)  {\footnotesize $\Sigma_{1}$};
                 \node at (4,-0.7)  {\footnotesize $\Sigma^*_{1}$};
                   \node at (-4,0.7)  {\footnotesize $\Sigma_{4}$};
                 \node at (-4,-0.7)  {\footnotesize $\Sigma^*_{4}$};
               \node at (1.4,0.7)  {\footnotesize $\Sigma_{2}$};
                 \node at (1.4,-0.7)  {\footnotesize $\Sigma^*_{2}$};
                \node at (-1.4,0.7)  {\footnotesize $\Sigma_{3}$};
                 \node at (-1.4,-0.7)  {\footnotesize $\Sigma^*_{3}$};

              \draw[dotted] (2.7,1.2) arc (450:90:1.2);
              \draw[dotted] (-2.7,1.2) arc (450:90:1.2);
              \node[]  at (2.7, 1.5) {\footnotesize $\partial \mathcal{U}_{r}$};
                \node[]  at (-2.7, 1.5) {\footnotesize $\partial \mathcal{U}_{l}$};


               \end{tikzpicture}
            \caption{ \footnotesize{ The jump contour $\Sigma^{loc}$ for $M^{loc}(z)$.  }}
      \label{scalingzk0}
        \end{center}
    \end{figure}
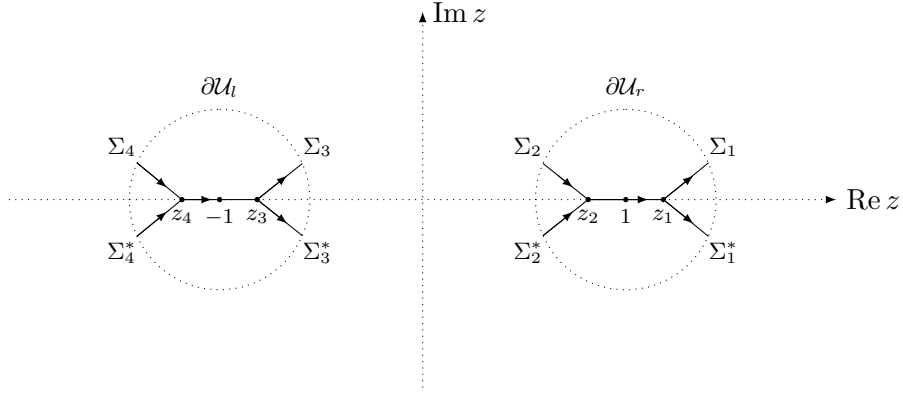

\begin{prob}\label{mloc}
    Find  $M^{loc}(z)=M^{loc}(z;x,t)$ which satisfies
	  \begin{itemize}
	  	\item  $M^{loc}(z)$ is analytic in $\mathbb{C}\setminus \Sigma^{loc}$.
	  	\item $M^{loc}(z)$  satisfies the jump condition
\begin{equation*}
	  		M^{loc}_+(z)=M^{loc}_-(z)V^{loc}(z), \ \ z \in \Sigma^{loc},
	  	\end{equation*}
	  	where $V^{loc}(z) := V^{(4)}(z)|_{ z \in \Sigma^{loc}}$.
	  	\item   $M^{loc}(z)$  has the same asymptotics with  $M^{rhp}(z)$.

	  \end{itemize}
\end{prob}

Based on the theorem of Beals-Coifman, we know as $t \to \infty$, the solution  $M^{loc}(z)$ is approximated by the sum of the separate local model in the neighborhood of $1$ and $-1$ respectively.

\begin{prob}\label{mr}
    Find  $ M^{j}(z)=M^{j}(z;x,t)$, $j\in \{r,l\}$ with properties
	  \begin{itemize}
	  	\item  $M^{j}(z)$ is analytic in $\mathbb{C} \setminus  {\Sigma}_{j}$  where  $\Sigma_{j}:= \Sigma^{(4)} \cap \mathcal{U}_{j}, \, j\in\{r,l\}$,
	  	\item  $M^{j}(z)$  satisfies the jump condition
\begin{equation*}
	  		M^{j}_+(z)=M^{j}_-(z) {V}^{j}(z), \ z\in  {\Sigma}_{j},
\end{equation*}
where ${V}^{j}(z)=V^{(4)}(z)|_{z\in \Sigma_j}, \, j \in \{r,l \}$.
\item As $z\to \infty$ in $\mathbb{C} \setminus  {\Sigma}_{j}$, $M^{j}(z)=I+\mathcal{O}(z^{-1})$.
\end{itemize}

\end{prob}

In the region $ -C < (\xi+6) t^{2/3}< 0$ with $C>0$,  we notice that
 $ \xi \to -6^-$ as $t \to  \infty$.
 From
(\ref{eq219}),  this leads to the coalescence of saddle points:
 $z_1$ and $z_2$  merge to $z=1$, while $z_3$ and $z_4$ collide at  $z=-1$.
The phase faction $t \theta(z)$ can be approximated with the help of scaled spectral variables:
\begin{itemize}
\item For $z$ close to  $1$,
\begin{equation}\label{tthe1}
 t \theta(z)  =  4t (z-1)^3 +(6t+x)(z-1) +A(z),
\end{equation}
where
\begin{equation}\label{tthea}
A(z) =  \frac{1}{2} \left( 3t(z-1)^2 -7t(z-1)^3 +  \sum_{n=2}^\infty  (-1)^{n+1}  \left( x+ \frac{n^2+3n+8}{2} t\right)  (z-1)^n  \right).
\end{equation}
Observing the characteristics of the above expansion,  we introduce  the following scaled spectral variables to match with the coefficients of the exponential terms in the Painlev\'{e} \uppercase\expandafter{\romannumeral2} model RH problem, as defined in \ref{appx}:

Define $s$ be the space-time parameter and  $\hat{k}$  be the scaled parameter
\begin{equation}\label{scaled1}
s = \frac{1}{3}(\xi+6)(3t)^{2/3}, \quad  \hat{k} = (3t)^{1/3} (z-1),
\end{equation}
then it can be proven that $A(z)$ converges and $A(z)=\mathcal{O}(t^{-1/3} \hat{k}^4)$.
Therefore, (\ref{tthe1}) becomes
\begin{equation}
t \theta(z) = \frac{4}{3} \hat{k}^3 + s \hat{k} + \mathcal{O}(t^{-1/3} \hat{k}^4).\label{asymn1}
\end{equation}

\item For $z$ close to  $-1$,
\begin{equation}\label{tthe2}
 t \theta(z)  =   4t(z+1)^3+ (6t+x)(z+1) +B(z),
\end{equation}
where
\begin{equation}
B(z) =  \frac{1}{2} \left( -3t(z+1)^2 -7t(z+1)^3 +  \sum_{n=2}^\infty \left( x+ \frac{n^2+3n+8}{2} t\right)  (z+1)^n  \right).
\end{equation}
Similar to the case for $z \to 1$, the space-time parameter $s$ is defined as  (\ref{scaled1}) and we define the scaled parameter $\check{k}$ as
\begin{equation}
  \check{k} = (3t)^{1/3} (z+1). \label{scaled2}
\end{equation}
Then, $B(z)$ converges and $B(z)= \mathcal{O}(t^{-1/3}\check{k}^4)$.
Therefore, (\ref{tthe2}) becomes
\begin{equation}
 t \theta(z)  =    \frac{4}{3} \check{k}^3 + s \check{k} + \mathcal{O}(t^{-1/3}\check{k}^4).\label{asymn2}
\end{equation}

\end{itemize}


\begin{remark}
 For the case of  the mKdV equation \eqref{dmkdv}   with  ZBCs,    the phase function is
\begin{equation}\label{phazbc}
t\theta(z)=4tz^3+xz.
\end{equation}
We carry out the following scaling:
$$ z \to k t^{-1/3},$$
 and (\ref{phazbc}) becomes
 \begin{equation*}
 t\theta(k t^{-1/3})=4k^3+xk t^{-1/3} = 3( \frac{4}{3}k^3 +sk),
 \end{equation*}
 where $s = \frac{1}{3}\xi t^{-2/3}$ with $\xi = x/t$. This indicates that, under this condition, the coefficients of the exponential terms in the local model can exactly match those of  the
   Painlev\'e II model RH problem.
However, in  the   case  of NZBCs  \eqref{dmkdv}-\eqref{bdries},  the phase function can only approximate  the exponential term coefficients of the Painlev\'e II  model RH problem with an error of $\mathcal{O}(t^{-1/3}k^4)$.

\end{remark}

Next we define two open disks associated with the scaled parameters $\hat{k} $ and $\check{k}$
  \begin{equation*}
 \widehat{\mathcal{U}}_r = \{\hat{k} \in \mathbb{C}: |\hat{k}|< (3t)^{\tau}\varepsilon \},\ \  \widecheck{\mathcal{U}}_l = \{\check{k} \in \mathbb{C}: |\check{k}|< (3t)^{\tau}\varepsilon \},
\end{equation*}
whose   boundaries are  oriented counterclockwise. Then  the transformation given by (\ref{scaled1})
defines a map $z \mapsto \hat{k}$, which maps  $ \mathcal{U}_{r}$  onto  $\widehat{ \mathcal{U}}_r$  in the $\hat{k}$-plane,
while the transformation given by (\ref{scaled2}) maps  $ \mathcal{U}_{l}$  onto $\widecheck{\mathcal{U}}_l$ in the $\check{k}$-plane.

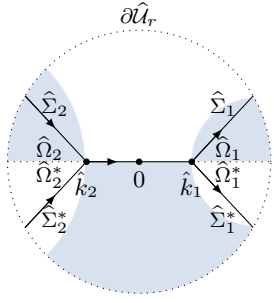
\begin{figure}
 \begin{center}
   \begin{tikzpicture}[scale=0.7]
    \draw[blue!10, fill=LightSteelBlue!50] (5,0) arc [start angle=180, end angle=90, radius =1.21]--(6.1651,1.25) arc [start angle=30, end angle=0, radius =2.5]--(6,0) --(5,0);

     \draw[blue!10, fill=LightSteelBlue!50] (1.5,0) arc [start angle=180, end angle=135, radius =2.5]--(2.23,1.77) arc [start angle=45, end angle=0, radius =2.5]--(3,0) --(1.5,0);
     \draw[blue!10, fill=LightSteelBlue!50] (1.5,0) arc [start angle=180, end angle=360, radius =2.5]--(6.5,0);

      \draw[white, fill=white] (5,0) arc [start angle=180, end angle=270, radius =1.21]--(6.1651,-1.25) arc [start angle=330, end angle=360, radius =2.5]--(6,0) --(5,0);
     \draw[white, fill=white] (1.5,0) arc [start angle=180, end angle=225, radius =2.5]--(2.23,-1.77) arc [start angle=315, end angle=360, radius =2.5]--(3,0) --(1.5,0);
      \draw[dotted                           ](4,0) circle (2.5);
        \draw[](1.8349,1.25)--(3,0);
        \draw [-latex] (1.8349,1.25)--(2.41745,0.625);
        \draw[](1.8349,-1.25)--(3,0);
        \draw [-latex] (1.8349,-1.25)--(2.41745,-0.625);
        \draw[](3,0)--(5,0);
        \draw[-latex](3,0)--(3.6,0);
        \node[shape=circle,fill=black, scale=0.13]  at (3,0){0};
        \node[below] at (3,0) {\footnotesize $\hat{k}_2$};
        \node[shape=circle,fill=black, scale=0.13]  at (4,0){0};
        \node[below] at (4,0) {\footnotesize $0$};
        \node[shape=circle,fill=black, scale=0.13]  at (5,0){0};
        \node[below] at (5,0) {\footnotesize $\hat{k}_1$};
        \draw[](5,0)--(6.1651,1.25);
        \draw [-latex] (5,0)--(5.58255,0.625);
        \draw[](5,0)--(6.1651,-1.25);
        \draw [-latex] (5,0)--(5.58255,-0.625);
        \draw [dotted] (-0.5, 0)--(3, 0);
        \draw [-latex,dotted] (5, 0)--(8.5, 0) node[right]{ \textcolor{black}{$\re \hat k$}};
           \draw [-latex,dotted] (4, -3)--(4, 4) node[right]{ \textcolor{black}{$\im \hat k$}};
        \node  at (2.4,1.1) {\scriptsize $\widehat{\Sigma}_{2}$};
        \node  at (2.4,-1.1) {\scriptsize $\widehat{\Sigma}^*_{2}$};
        \node  at (5.6,1.1) {\scriptsize $\widehat{\Sigma}_{1}$};
        \node  at (5.6,-1.1) {\scriptsize $\widehat{\Sigma}^*_{1}$};
        \node  at (2.3,0.25) {\scriptsize $\widehat{\Omega}_{2}$};
        \node  at (2.3,-0.25) {\scriptsize $\widehat{\Omega}^*_{2}$};
        \node  at (5.7,0.25) {\scriptsize $\widehat{\Omega}_{1}$};
        \node  at (5.7,-0.25) {\scriptsize $\widehat{\Omega}^*_{1}$};
        \node[]  at (3.2,2.8) {\scriptsize $\partial \widehat{\mathcal{U}}_{r}$};

    \end{tikzpicture}
    \caption{\footnotesize{The jump contour $\widehat{\Sigma}_r$ for $M^r(\hat{k})$ on  $\widehat{\mathcal{U}}_r$.  }  }
      \label{scalingzk}
  \end{center}
\end{figure}

First, we construct the local parametrix $M^r(z)$.
Define the contour $\widehat{\Sigma}_r$ in the $\hat{k}$-plane
\begin{equation*}
\widehat{\Sigma}_r := \bigcup_{j=1,2} \left( \widehat{\Sigma}_j \cup \widehat{\Sigma}^*_j \right)  \cup (\hat{k}_2,\hat{k}_1),
\end{equation*}
which corresponds to  $\Sigma_r$ after scaling $z$ to the scaled parameter $\hat{k}$. The corresponding regions in the $\hat{k}$-plane can be  seen in  Figure \ref{scalingzk}.
Correspondingly, the saddle points $z_1$ and $z_2$ in the $z$-plane are rescaled to $\hat{k}_1$ and $\hat{k}_2$ respectively in  the $\hat{k}$-plane with $\hat{k}_j = (3t)^{1/3}(z_j-1),\, j=1,2$.
Moreover, (\ref{scaled1}) also reveals that
\begin{equation*}
z = (3t)^{-1/3} \hat{k}+1.
\end{equation*}
Therefore, the jump matrix $V^r(z)$ transforms to  the following $V^r(\hat{k})$ in the $\hat{k}$-plane
\begin{align*}
        V^{r}(\hat{k})= \begin{cases}
       e^{it\theta\left( (3t)^{-1/3} \hat{k}+1\right)\widehat\sigma_3  }   \left( \begin{array}{cc}
       		1& -\gamma(z_j)\\
       		 0   & 1
       	\end{array}\right),\  \hat{k} \in \widehat{\Sigma}_{j}, \ j=1,2, \\
  e^{it\theta\left( (3t)^{-1/3} \hat{k}+1\right)\widehat\sigma_3  }  \left(	\begin{array}{cc}
       			1& 0\\
       			\overline{\gamma(z_j)} & 1
       		\end{array}\right) ,\ \hat{k} \in \widehat{\Sigma}^*_{j}, \ j=1,2,\\
 \delta_-\left( (3t)^{-1/3} \hat{k}+1\right)^{-\sigma_3} V\left( (3t)^{-1/3} \hat{k}+1\right)\delta_+\left( (3t)^{-1/3} \hat{k}+1\right)^{\sigma_3},  \   \hat{k}\in [\hat{k}_2,  \hat{k}_1].
       	\end{cases}
       \end{align*}

In the generic case, $z_j \to 1$ and $r(z_j) \to -i$ for $j=1,2$ as $t \to \infty$, which causes the appearance of the singularity of $\frac{r(z_j)}{1-|r(z_j)|^2}$. However, this singularity can be balanced by the factor $\delta(z)^{-2}$.
 Define a cutoff function $\chi(z)\in C_0^\infty(\mathbb{R},[0,1])$
satisfying
\begin{equation}
\chi (z)=1, \ z\in  \mathbb{R}\cap \mathcal{U}_r,
\end{equation}
and  a new  reflection coefficient  $\tilde{r}(z)$  satisfying
\begin{equation}
\tilde{r}(z)=(1-\chi (z)) \frac{\overline{r(z)} }{1-|r(z)|^2}  \delta _+( z)^{-2}+\chi (z)  f(z)  h(z)^2,
\label{Rz2}
\end{equation}
where
$$f(z) := \frac{ \overline{S_{21} (z)} }{S_{11} (z)},  \ \ \  h(z):=\frac{ a(z) }{\delta_+(z)}, $$
and $S_{21}(z)$ and  $S_{11}(z)$ are defined by (\ref{S21}) and (\ref{S11}) respectively,
while in the non-generic case,
\begin{equation}
\tilde{r}(z)= \frac{\overline{r(z)} }{1-|r(z)|^2}  \delta_+ ( z)^{-2}.
\label{Rz22}
\end{equation}
Moreover, we have $\tilde{r}(z_j)=\gamma(z_j)$ as $z_j \to 1, \ j=1,2$.

Next we will show that  in  $\widehat{\mathcal{U}}_{r}$,   RH problem for $M^r(\hat{k})$ can be explicitly approximated by the following model RH problem for  $\widehat{M}^{r}(\hat{k})$, and then prove the solution $\widehat{M}^{r}(\hat{k})$ is associated to the Painlev\'e II equation.
\begin{prob}
    Find  $ \widehat{M}^{r}(\hat{k})=\widehat{M}^{r}(\hat{k};x,t)$ with properties
	  \begin{itemize}
	  	\item  $\widehat{M}^{r}(\hat{k})$ is analytic in $\mathbb{C} \setminus  \widehat{\Sigma}_{r}$.
	  	\item  $\widehat{M}^{r}(\hat{k})$  satisfies the jump condition
\begin{equation*}
	  		\widehat{M}^{r}_+(\hat{k})=\widehat{M}^{r}_-(\hat{k}) \widehat{V}^{r}(\hat{k}), \ \hat{k}\in  \widehat{\Sigma}_{r},
\end{equation*}
\end{itemize}
where
\begin{align*}
        \widehat{V}^{r}(\hat{k})= \begin{cases}
       e^{i \left( \frac{4}{3} \hat{k}^3 +  s \hat{k} \right)\widehat\sigma_3  }   \left( \begin{array}{cc}
       		1& -\tilde{r}(1)\\
       		 0   & 1
       	\end{array}\right):=b_+^{-1},\  \hat{k}\in \widehat{\Sigma}_{j}, \ j=1,2, \\
  e^{i \left( \frac{4}{3} \hat{k}^3 +  s \hat{k} \right)\widehat\sigma_3  }  \left(	\begin{array}{cc}
       			1& 0\\
       			\overline{\tilde{r}(1)} & 1
       		\end{array}\right):=b_- ,\ \hat{k} \in \widehat{\Sigma}^*_{j}, \ j=1,2,\\
b_-b_+^{-1},  \   \hat{k}\in [\hat{k}_2,  \hat{k}_1].
       	\end{cases}
       \end{align*}

	  \begin{itemize}
         \item   $\widehat{M}^{r}(\hat{k})\to I, \ \hat{k}\to \infty$.
	  \end{itemize}
\end{prob}

Define $N(\hat{k}):=M^r(\hat{k}) (\widehat{M}^r(\hat{k}))^{-1}$ which satisfies the following RH problem.
\begin{prob}
 Find  $ N(\hat{k})=N(\hat{k};x,t)$  such that
 \begin{itemize}
   \item $N(\hat{k})$ is analytic in $\mathbb{C} \setminus  \widehat{\Sigma}_{r}$.
   \item  $N(\hat{k})$ satisfies the following jump condition
   \begin{equation*}
     N_+(\hat{k}) = N_-(\hat{k})V^N(\hat{k}),
   \end{equation*}
   where
   \begin{equation*}
   V^N(\hat{k}) = \widehat{M}_-^r(\hat{k}) V^r(\hat{k}) \widehat{V}^r(\hat{k})^{-1} \widehat{M}_-^r(\hat{k})^{-1}.
   \end{equation*}
 \end{itemize}
\end{prob}

\begin{proposition}\label{locpain}
As $t \to \infty$, $N(\hat{k})$ exists and satisfies
\begin{equation}\label{esN}
 N(\hat{k}) = I + \mathcal{O}(t^{-\frac{1}{3}+4\tau}),
\end{equation}
where $\tau$ is a constant with $0<\tau<\frac{1}{30}$.
\end{proposition}

\begin{proof}
 Suppose that $\widehat{M}^r(\hat{k})$ is bounded, which we will show in (\ref{pre1}) and (\ref{mrt}), we only need to estimate the error between $V^r(\hat{k})$ and $\widehat{V}^r(\hat{k})$.
 For $\hat{k} \in \left(\hat{k}_2, \hat{k}_1\right)$, $\left|  e^{2i t   \theta (z) } \right|=\left| e^{  i(\frac{8}{3} \hat{k}^3+ 2s \hat{k})  } \right|=1$.
 Direct calculations show that
\begin{align}
&\left| \tilde{r}(z)   e^{2i t   \theta \left(z\right) } - \tilde{r}(1)  e^{  i(\frac{8}{3} \hat{k}^3+ 2s \hat{k})  }    \right| \nonumber\\
&\leq \left| \tilde{r}(z) - \tilde{r}(1) \right|
+ \left|   \tilde{r}(1) \right|       \left|   e^{2i t   \theta \left(z\right) }-   e^{  i(\frac{8}{3} \hat{k}^3+ 2s \hat{k})  } \right|\nonumber\\
& \leq \left|h(z)\right|^2\left|f(z)-f(1) \right| +\left|f(1)\right| \left|h^2(z)-h^2(1) \right|\nonumber\\
 &+ \left|   \tilde{r}(1) \right|       \left|   e^{2i t   \theta \left(z\right) }-   e^{  i(\frac{8}{3} \hat{k}^3+ 2s \hat{k})  } \right|.\label{poe1}
 \end{align}
Following the idea of Proposition 3.2 in \cite{zx1},  with the H\"{o}lder  inequality, we have
  \begin{align}
& \left| f(z) - f(1) \right|  \leq  \|r \|_{H^1(\mathbb{R}) } |z-1|^{\frac{1}{2}}     t^{-\frac{1}{6}}\lesssim t^{-\frac{1}{3}+\frac{\tau}{2}}, \label{poe2}\\
&  \left|   e^{i\mathcal{O}(t^{-\frac{1}{3}} \hat{k}^4)}-   1  \right| \leq e^{|\mathcal{O}(t^{-\frac{1}{3}} \hat{k}^4)|}- 1    \lesssim t^{-\frac{1}{3}+4\tau}. \label{poe3}
 \end{align}
 From (\ref{traceformula}) and (\ref{delta}), it is straightforward to check that
 \begin{equation}
 h(z) = \prod_{n=1}^{2N}\frac{z-\eta_n}{z-\bar{\eta}_n},
 \end{equation}
 then,
  \begin{equation}\label{hzh1}
 |h(z)-h(1)| \leq \left| \prod_{n=1}^{2N}\frac{1-\eta_n}{1-\bar{\eta}_n}\right| \left|  \prod_{n=1}^{2N}\frac{z-\eta_n}{z-\bar{\eta}_n} \frac{1-\bar{\eta}_n}{1-\eta_n}-1\right| \lesssim t^{-\frac{2}{3} + 2\tau}.
  \end{equation}
 Substituting (\ref{poe2}), (\ref{poe3}), and (\ref{hzh1}) into (\ref{poe1}),  we obtain
 \begin{equation}\label{esr1}
 \Big| \tilde{r}(z)  e^{2it\theta \left(z\right)}- \tilde{r}(1)   e^{i \left( \frac{8}{3} \hat{k}^3 + 2 s \hat{k} \right)}  \Big|   \lesssim t^{-\frac{1}{3}+4\tau}.
 \end{equation}
 For $\hat{k} \in  \widehat{\Sigma}_{1}$,  denote $\hat{k} = \hat{k}_1 +\hat{u}+i\hat{v}$ and $z = z_1 + u+iv$.
 To prove that $\re i\left(\frac{8}{3}\hat{k}^3 +2s \hat{k}\right) \le - \frac{16}{3}\hat{u}^2 \hat{v}$ holds, we only need to prove the following inequality holds
 \begin{equation}
 \re i\left(8t(z-1)^3 +2(x+6t)(z-1)\right) \le - 16 u^2v.
 \end{equation}
From (\ref{xixi1}), it is easy to infer that
 \begin{equation}
 \re i\left(8t(z-1)^3 +2(x+6t)(z-1)\right) \le -16tu^2v -2tv w(z_1),
  \end{equation}
where
\begin{equation}
w(z_1) =12(z_1-1)^2+24(z_1-1)u-\frac{3 \left( z_1^4+1 \right)}{z_1^2}+6.
\end{equation}
Since $w'(z_1) \ge 0$ on the interval $[1,+\infty)$, $w(z_1) \ge w(1)=0$ and then
 \begin{equation}
 \re i\left(8t(z-1)^3 +2(x+6t)(z-1)\right) \le -16tu^2v.
  \end{equation}
Therefore, $ \left| e^{i \left( \frac{8}{3} \hat{k}^3 + 2 s \hat{k} \right)}\right|$ is bounded.
Similarly to the case on the real axis, we can obtain
\begin{equation}\label{esrz}
\Big|\tilde{r}(z_j)   e^{2it\theta \left(z\right)}-\tilde{r}(1) e^{i \left( \frac{8}{3} \hat{k}^3 + 2 s \hat{k} \right) }  \Big|   \lesssim t^{-\frac{1}{3}+4\tau},    \ \hat{k} \in  \widehat{\Sigma}_{1}.
\end{equation}
The estimate on other jump contours can be given in a similar way. (\ref{esr1}) and (\ref{esrz}) implies that $\Vert V^N-I \Vert_{L^1 \cap L^2 \cap L^\infty} \lesssim t^{-\frac{1}{3}+4\tau}$ uniformly. Therefore, the existence and uniqueness of $N(\hat{k})$ can be proven by the  theorem of the small-norm RH problem \cite{PX3}, which also yields (\ref{esN}).

\end{proof}

Therefore, the solution  $ \widehat{M}^{r}(\hat{k})$ is crucial to our analysis.  Next we show it is related with the Painlev\'{e} II equation via an  appropriately equivalent  deformation.
For this purpose, we add four auxiliary lines $L_j,\, j=1,2,3,4$  passing through the point $\hat{k}=0$
at  the angle  $\pi/3$  with real axis,
which together with the original contour  $\widehat{\Sigma}_r$
 divide the complex plane into 10 regions $\widetilde{\Omega}_j,\, j=1,\cdots,6$ and $\widehat{\Omega}_j \cup \widehat{\Omega}^*_j,\, j=1,2$. See Figure \ref{desc45}.

We further define
\begin{align}
&P(\hat{k}) =\left\{\begin{matrix}
b_+,\  \ & \hat{k}\in \widetilde{\Omega}_1\cup\widetilde{\Omega}_3,\cr
b_-,\  \ & \hat{k}\in \widetilde{\Omega}_4\cup\widetilde{\Omega}_6,\cr
I,\  \ & \text{elsewhere},
\end{matrix}\right. \nonumber
\end{align}
and  make a transformation
\begin{equation}
\widetilde M^{r} (\hat{k})=\widehat M^{r} (\hat{k})P(\hat{k}),\label{pre1}
\end{equation}
then we obtain the following RH problem.

\begin{prob}\label{mp2}
    Find   $\widetilde M^{r} (\hat{k})=\widetilde M^{r} (\hat{k};s )$ with properties
    \begin{itemize}
        \item  $\widetilde M^{r} (\hat{k})$ is analytic in $\mathbb{C} \setminus  \widetilde \Sigma^{P} $, where $\widetilde \Sigma^{P}=\cup_{j=1}^4 L_{j}. $ See  Figure  \ref{desc45}.
        \item  $\widetilde M^{r}( \hat{k})$ satisfies the  jump condition \begin{equation*}
       \widetilde M^{r}_+( \hat{k})=\widetilde M^{r}_-(\hat{k}) \widetilde V^{P}(\hat{k}), \ \ \hat{k}\in \widetilde \Sigma^{P},
        \end{equation*}
        where
       \begin{align}\label{vp}
   \widetilde V^{P}(\hat{k})= \begin{cases}
               b_+^{-1},\quad \hat{k}\in L_{1}\cup L_{2},\\
               b_-, \quad \hat{k} \in L_{3} \cup L_{4}.
              	\end{cases}
       \end{align}

        \item   $\widetilde{M}^{r}(\hat{k})\to I, \ \hat{k}\to \infty$.

\end{itemize}
    \end{prob}

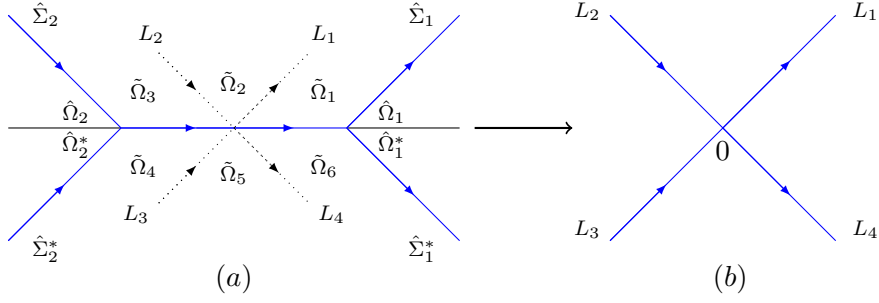
\begin{figure}
\begin{center}
\begin{tikzpicture}\label{Fig3}

\draw [](-6.5,0)--(-5,0);
\draw [](-0.5,0)--(-2,0);
\draw [ ](-5,0)--(-2,0);
 \draw [    ](-5,0)--(-6.5,1.5);
 \draw [   -latex](-5,0)--(-4,0);
\draw [  -latex](-4,0)--(-2.7,0);
  \draw [  -latex ] (-6.5,1.5)--(-5.75,1.5/2);
 \draw [  ](-5,0)--(-6.5,-1.5);
   \draw [  -latex ] (-6.5,-1.5)--(-5.75,-1.5/2);
  \draw [  ](-2,0)--(-0.5,1.5 );
     \draw [  -latex ] (-2,0)--(-1.1, 0.9);
 \draw[  ](-2,0)--(-0.5,-1.5);
      \draw [  -latex ] (-2,0)--(-1.1,-0.9);

 \draw[dotted ](-4.5,-1)--(-2.5, 1 );
 \draw[dotted,-latex ](-4.5,-1)--(-4, -0.5 );
 \draw[dotted,-latex ](-3.5,0)--(-2.9,  0.6 );
  \draw[dotted ](-4.5, 1)--(-2.5,-1 );
   \draw[dotted,-latex ](-4.5, 1)--(-4,  0.5 );
   \draw[dotted,-latex ](-3.5,0)--(-2.9, -0.6 );

  \node    at (-1.4, 0.22)  {\scriptsize $\hat{\Omega}_1$};
    \node    at (-1.4, -0.22)  {\scriptsize $\hat{\Omega}^*_1$};
     \node    at (-5.6, 0.22)  {\scriptsize $\hat{\Omega}_2$};
  \node    at (-5.6, -0.22)  {\scriptsize $\hat{\Omega}^*_2$};
    \node    at (-2.3, 0.5)  {\scriptsize $\tilde{\Omega}_1$};
 \node    at (-3.5,0.6)  {\scriptsize $\tilde{\Omega}_2$};
 \node    at (-4.7, 0.5)  {\scriptsize $\tilde{\Omega}_3$};

  \node    at (-4.7,-0.5)  {\scriptsize $\tilde{\Omega}_4$};
  \node    at (-3.5, -0.6)  {\scriptsize $\tilde{\Omega}_5$};
    \node    at (-2.3, -0.5)  {\scriptsize $\tilde{\Omega}_6$};
     \node  [below]  at (-1,1.8) {\scriptsize $\hat{\Sigma}_1$};
     \node  [below]  at (-6,1.8) {\scriptsize $\hat{\Sigma}_2$};
   \node  [below]  at (-6,-1.3) {\scriptsize $\hat{{\Sigma}}^*_2$};
       \node  [below]  at (-1,-1.3) {\scriptsize $\hat{\Sigma}^*_1$};

    \node  [below]  at (-2.3,1.5) {\scriptsize $L_1$};
        \node  [below]  at (-4.6,1.5) {\scriptsize $L_2$};

 \node  [below]  at (-4.8,-0.9) {\scriptsize $L_3$};
  \node  [below]  at (-2.2,-0.9) {\scriptsize $L_4$};


\draw[      ](1.5,-1.5)--(4.5,1.5);
\draw[  -latex ](1.5,-1.5)--(2.25,-0.75);
\draw[ -latex ](3, 0)--(3.9, 0.9);
\draw[    ](1.5,1.5)--(4.5,-1.5);
\draw[  -latex ](1.5, 1.5)--(2.25, 0.75);
\draw[  -latex ](3, 0)--(3.9,-0.9);
\node  [below]  at (1.2,1.8) {\scriptsize $L_2$};
\node  [below]  at (4.9,1.8) {\scriptsize $L_1$};
\node  [below]  at (4.9,-1.1) {\scriptsize $L_4$};
\node  [below]  at (1.2,-1.1) {\scriptsize $L_3$};
\node    at (3,-0.3)  {$0$};
 \draw [thick,->  ](-0.3,0)--(1,0);
  \node    at (-3.5, -2 )  {  $ (a)$};
    \node    at ( 3.1, -2 )  {  $ (b)$};
\end{tikzpicture}
\end{center}
\caption {\small   In plot(a),   four new dashed auxiliary lines are added to the jump contour  of $\widehat{M}^{r} (\hat{k})$, which then can be deformed
into the  jump contour of the Painlev\'e II model RH problem for $ M^{P}(\hat{k})$, as shown in plot (b).}
\label{desc45}
\end{figure}

Let $$\varphi_0=\arg \tilde{r}(1) .$$  Then $$\tilde{r}(1)= |\tilde{r}(1)| e^{i\varphi_0}.$$
 Following the idea  \cite{Fokas1}, we find the solution  $\widetilde M^{r}(\hat{k})$ can be expressed by
\begin{equation}\label{mrt}
   \widetilde M^{r}(\hat{k}) = \sigma_1 e^{- i  \left( \frac{\pi}{4}+\frac{\varphi_0}{2} \right)\widehat\sigma_3} M^P(\hat{k})\sigma_1,
\end{equation}
where $M^P(\hat{k})$  satisfies a standard Painlev\'{e}  II  model RH problem given in \ref{appx}  with the parameter $q=i|\tilde{r}(1)|$.
By Proposition \ref{locpain},  we have
\begin{align}\label{mr1}
     M^{r}_1(s) = \widetilde M^{r}_1(s)  + \mathcal{O}(t^{-\frac{1}{3}+4\tau}) =   \sigma_1 e^{- i  \left( \frac{\pi}{4}+\frac{\varphi_0}{2} \right) \widehat\sigma_3} M^P_1(s) \sigma_1 + \mathcal{O}(t^{-\frac{1}{3}+4\tau}) ,
\end{align}
where   the subscript ``$1$'' represents the coefficient  of the  $\hat{k}^{-1}$ term in the asymptotic expansion of the corresponding solution as $\hat{k} \to \infty$,
and $M^P_1(s)$ is given by (\ref{posee}).

 Finally, as $\hat{k}\to\infty$, $M^{r}(z)$ can be described by the following equation
\begin{align}\label{pre2}
M^{r}(z) = I + \frac{ M^{r}_1(s)}{(3t)^{1/3}(z-1)} + \mathcal{O}(t^{-\frac{2}{3}+2\tau}).
\end{align}
A similar process gives the solution $M^l(z)$,  which has the asymptotics: as $\check{k} \to \infty$,
\begin{equation}
M^l(z) = I + \frac{M^l_1(s)}{(3t)^{1/3}(z+1)} + \mathcal{O}(t^{-\frac{2}{3}+2\tau}),
\end{equation}
where
\begin{align}\label{ml1}
M^l_1(s) = -\sigma_1 M^r_1(s) \sigma_1.
\end{align}

Now we construct the solution $M^{rhp}(z)$ of the form
 \begin{align}\label{trans6}
 M^{rhp}(z) =
  \begin{cases}
   E(z), \ z \in \mathbb{C} \setminus \left( \mathcal{U}_r \cup \mathcal{U}_l\right),\\
   E(z) M^{r}(z), \ z \in \mathcal{U}_r,\\
   E(z) M^l(z), \ z \in \mathcal{U}_l.
  \end{cases}
 \end{align}
Since  the solutions $M^j(z), \, j \in \{r,l\}$  have been obtained, we can construct the solution $M^{rhp}(z)$ if we find the error function $E(z)$.

\subsubsection{A small-norm RH problem}\label{error}
We   consider the error function  $E(z)$ defined by (\ref{trans6}), which admits the following  RH problem.

\begin{figure}[htbp]
    \begin{center}
        \begin{tikzpicture}

               \draw[-latex,dotted](-5.5,0)--(5.5,0)node[right]{ \textcolor{black}{$\re z$}};
               \draw[-latex,dotted](0,-2.21)--(0,2.21)node[right]{\textcolor{black}{$\im z$}};

               \draw[teal] (3.8,0) arc (0.5:360:1.1);
               \draw[teal] (-1.6,0) arc (0.5:360:1.1);
               \coordinate (I) at (0,0);
               \coordinate (A) at (-3.3,0);
               \fill (A) circle (0pt) node[below] {\footnotesize $\xi_4$};
               \coordinate (b) at (-2.1,0);
               \fill (b) circle (0pt) node[below] {\footnotesize $\xi_3$};
               \coordinate (e) at (3.3,0);
               \fill (e) circle (0pt) node[below] {\footnotesize $\xi_1$};
               \coordinate (f) at (2.1,0);
               \fill (f) circle (0pt) node[below] {\footnotesize $\xi_2$};

               \coordinate (c) at (-2.65,0);
               \fill[] (c) circle (1pt) node[below] {\scriptsize$-1$};
               \coordinate (d) at (2.65,0);
               \fill[] (d) circle (1pt) node[below] {\scriptsize$1$};
               \draw[](0,0)--(1,0.8);
               \draw[-latex](0,0)--(0.6,0.48);
               \draw[](1,0.8)--(2,0);
               \draw[-latex](1,0.8)--(1.5,0.4);

               \draw[](3.5,0)--(5.5, 1.2);
               \draw[-latex](3.5,0)--(4.5,0.6);
               \draw[](0,0)--(1,-0.8);
               \draw[-latex](0,0)--(0.6,-0.48);
               \draw[](1,-0.8)--(2,0);
               \draw[-latex](1,-0.8)--(1.5,-0.4);

               \draw[](3.5,0)--(5.5, -1.2);
               \draw[-latex](3.5,0)--(4.5,-0.6);
               \draw[](0,0)--(-1,0.8);
               \draw[-latex](-1,0.8)--(-0.4,0.32);
               \draw[](-1,0.8)--(-2,0);
               \draw[-latex](-2,0)--(-1.3,0.56);

               \draw[](-3.5,0)--(-5.5, 1.2);
               \draw[ latex-](-4.5,0.6)--(-5.5, 1.2);
               \draw[](0,0)--(-1,-0.8);
               \draw[-latex](-1,-0.8)--(-0.4,-0.32);
               \draw[](-1,-0.8)--(-2,0);
               \draw[-latex](-2,0)--(-1.3,-0.56);

               \draw[](-3.5,0)--(-5.5, -1.2);
               \draw[ latex-](-4.5,-0.6)--(-5.5, -1.2);

               \draw[] (1,0)--(1,0.8);
                \draw[-latex] (1,0)--(1,0.4);
                 \draw[] (-1,0)--(-1,0.8);
                \draw[-latex] (-1,0)--(-1,0.4);
                \draw[] (1,0)--(1,-0.8);
                \draw[ -latex] (1,0)--(1,-0.4);
                 \draw[] (-1,0)--(-1,-0.8);
                \draw[-latex] (-1,0)--(-1,-0.4);
               \node at (4.6,0.9)  {\scriptsize $\Sigma_{1}$};
                 \node at (4.6,-0.9)  {\scriptsize $\overline{\Sigma}_{1}$};
                   \node at (-4.6,0.9)  {\scriptsize  $\Sigma_{4}$};
                 \node at (-4.6,-0.9)  {\scriptsize $\overline{\Sigma}_{4}$};
               \node at (1.5,0.7)  {\scriptsize $\Sigma_{2}$};
                 \node at (1.5,-0.7)  {\scriptsize $\overline{\Sigma}_{2}$};
                \node at (-1.5,0.7)  {\scriptsize $\Sigma_{3}$};
                 \node at (-1.5,-0.7)  {\scriptsize $\overline{\Sigma}_{3}$};
                  \node at (0.6,0.7)  {\scriptsize $\Sigma_{0^+}$};
                 \node at (0.6,-0.7)  {\scriptsize $\overline{\Sigma}_{0^+}$};
                \node at (-0.4,0.7)  {\scriptsize $\Sigma_{0^-}$};
                 \node at (-0.4,-0.7)  {\scriptsize $\overline{\Sigma}_{0^-}$};
                  \node at (0.7,0.25)  {\scriptsize $\Sigma'_{1}$};
                 \node at (1.25,-0.25)  {\scriptsize $\overline{\Sigma}'_{1}$};
                 \node at (-0.7,0.25)  {\scriptsize $\Sigma'_{2}$};
                 \node at (-1.25,-0.25)  {\scriptsize $\overline{\Sigma}'_{2}$};
                 \draw[white!20, fill=white] (-2.7,0) circle (1.1);
          \node[shape=circle,fill=black, scale=0.15]  at (-2.7,0){0} ;
        \draw[white!20, fill=white] (2.7,0) circle (1.1);
            \node[shape=circle,fill=black, scale=0.15]  at (2.7,0){0} ;
              \draw[-latex  ] (2.7,1.1) arc (450:90:1.1);
              \draw[-latex ] (-2.7,1.1) arc (450:90:1.1);
              \node[]  at (2.7, 1.4) {\scriptsize $\partial \mathcal{U}_{r}$};
                \node[]  at (-2.7, 1.4) {\scriptsize $\partial \mathcal{U}_{l}$};
                  \node at (2.7,-0.3)  {\footnotesize $1$};
                   \node at (-2.7,-0.3)  {\footnotesize $-1$};
               \end{tikzpicture}
                           \caption{\footnotesize{The  jump contour $\Sigma^E$  for $E(z)$. }}
                             \label{signdbiiiar}
        \end{center}
    \end{figure}
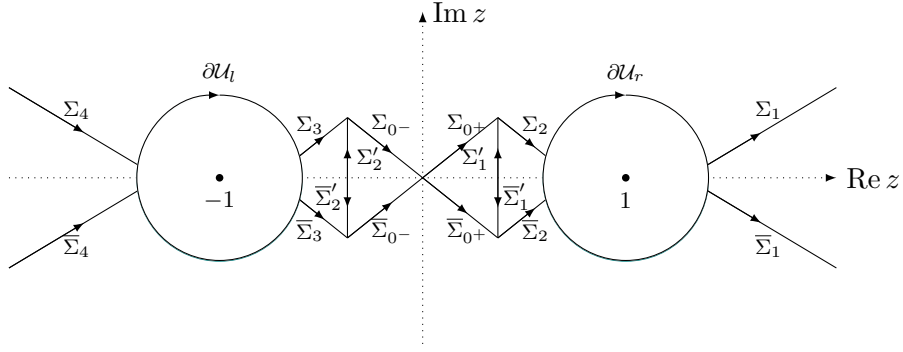

\begin{prob}\label{iew}
  Find      $E(z)$ with the properties
          \begin{itemize}
          	\item   $E(z)$ is analytic in $\mathbb{C}\setminus  \Sigma^{E}$,
          where $\Sigma^E = \left(\partial \mathcal{U}_r \cup \partial \mathcal{U}_l \right) \cup \left(\Sigma^{(4)} \setminus \left(\mathcal{U}_r \cup \mathcal{U}_l \right) \right)$. See Figure \ref{signdbiiiar}.
          	\item  $E(z)$  satisfies the jump condition
          	\begin{align*}
          		E_{+}(z)=E_-(z)V^{E}(z), \quad z\in \Sigma^{E},
          	\end{align*}
          	where the jump matrix $V^{E}(z)$ is given by
           \begin{align}
                V^{E}(z)= \begin{cases}
                   V^{(4)}(z),\quad z \in \Sigma^{(4)} \setminus \left(\mathcal{U}_r \cup \mathcal{U}_l \right),\\
                  M^{r}(z) , \quad z \in \partial \mathcal{U}_r,\\
                  M^{l}(z) , \quad z \in \partial \mathcal{U}_l. \label{ioei}
                \end{cases}
            \end{align}

          	\item       $E(z)=I+\mathcal{O}(z^{-1}),	\quad  z\to  \infty.$	

          \end{itemize}
\end{prob}

To obtain the existence of the solution $E(z)$,  we estimate the jump matrix $V^{E}(z)-I$.
\begin{proposition} Let $2 \le p \le \infty$ and $0<\tau<\frac{1}{30}$. It follows that
   \begin{equation}\label{vee}
                ||V^{E}(z)-I||_{L^p(\Sigma^E)}=\begin{cases}
                    \mathcal{O}(e^{-c t^{3\tau}}),\quad z\in \Sigma^{E} \setminus \left(\mathcal{U}_r \cup \mathcal{U}_l \right), \\
                    \mathcal{O}(t^{-\kappa_p}), \quad z\in \partial \mathcal{U}_r \cup \partial \mathcal{U}_l,
                \end{cases}
            \end{equation}
            for some positive constant $c$ with $\kappa_\infty=\tau$ and $\kappa_2=\frac{1}{6}+\frac{\tau}{2}$.
\end{proposition}
\begin{proof}
First, we consider the estimate when $p=\infty$.
For $z\in \Sigma^{E} \setminus \left(\mathcal{U}_r \cup \mathcal{U}_l \right)$, by (\ref{ioei}) and Proposition \ref{reprop1},
\begin{equation}
|V^{E}(z)-I| = |V^{(4)}(z)-I| \lesssim   e^{-c t^{3\tau}},
\end{equation}
where $c = c(\phi,\xi)$.
For $z \in \partial\mathcal{U}_r $, by (\ref{ioei}),
\begin{equation}
|V^{E}(z)-I| = |M^{r}(z)-I| \lesssim  t^{-\tau}.
\end{equation}
So does the estimate on $\partial\mathcal{U}_l$. Other cases when $2 \le p <\infty$ can be proven similarly.
\end{proof}

     This proposition establishes RH problem  \ref{iew} as a small-norm RH problem, for which there exists a well-known existence and uniqueness theorem \cite{PX3}.
     Define the integral operator $ C_{w^{E}}: L^2(\Sigma^E) \to L^2(\Sigma^E)$ by
        \begin{align*}
           C_{w^{E}}f =C_-\left( f \left( V^{E}(z)-I \right) \right),
        \end{align*}
         where $w^{E} =V^{E}(z)-I$ and  $C_-$ is the  Cauchy projection operator on $\Sigma^{E}$.
        By (\ref{vee}), a simple calculation shows that
       $$||C_{w^{E}}||_{L^2_{op}(\Sigma^{E})}\lesssim  ||C_-||_{L^2_{op}(\Sigma^{E})}||V^{E}(z)-I||_{L^\infty(\Sigma^{E})} \lesssim t^{-\tau}.$$
 According to the theorem of Beals-Coifman \cite{BC84},   the solution of RH problem \ref{iew} can be expressed in terms of
         \begin{equation}
            E(z)=I +\frac{1}{2\pi i} \int_{\Sigma^{E}} \frac{\mu_E(\zeta) \left( V^{E}(\zeta)-I\right)}{\zeta-z}\, \mathrm{d}\zeta, \nonumber
         \end{equation}
        where $\mu_E -I \in L^2\left( \Sigma^{E}\right)$   satisfies $\left(1-C_{w^{E}}\right)(\mu_E-I)=C_{w^{E}}I$.
  Furthermore, from (\ref{vee}), we have the estimates
       \begin{equation}
        ||V^{E}(z)-I||_{L^2(\Sigma^{E})}\lesssim t^{-\frac{1}{6}-\frac{\tau}{2}},\quad
          ||\mu_E-I||_{L^2(\Sigma^{E})} \lesssim  t^{-\frac{1}{6}-\frac{\tau}{2}}, \label{owej}
       \end{equation}
which imply that  RH problem \ref{iew} exists a unique solution.
On the other hand,   $\mu_E$ can be rewritten as  $$\mu_E=I+\sum_{j=1}^4 C_{w^E}^jI+(1-C_{w^{E}})^{-1}(C_{w^E}^5I),$$ where for $j=1,\dots,4$, we have the estimates
\begin{equation*}
    || C_{w^E}^jI||_{L^2(\Sigma^E)} \lesssim t^{-(\frac{1}{6}+j\tau-\frac{\tau}{2})}, \; ||\mu_E-I- \sum_{j=1}^4 C_{w^E}^jI||_{L^2(\Sigma^E)} \lesssim t^{-(\frac{1}{6}+9\tau)}.
\end{equation*}

To recover the potential $q(x, t)$ of (\ref{sol}), the behavior of $E(z)$ at $z=0$ and $z\to\infty$ must be characterized.
 We  first derive the expansion of  $E(z)$ as $z\to\infty$
        \begin{equation}
            E(z)=I + z^{-1}E_1 +\mathcal{O}\left(z^{-2}\right), \label{trans8}
        \end{equation}
      where
      $$ E_1=-\frac{1}{2\pi i} \int_{\Sigma^{E}} \mu_E(\zeta)\left( V^{E}(\zeta)-I \right)\, \mathrm{d}\zeta.$$

   \begin{proposition} \label{error1}
   $E_1$ and $E(0)$ can be estimated as follows:
	\begin{align}
		& E_1=    (3t)^{-\frac{1}{3}} \begin{pmatrix}i\int_s^\infty u(\zeta)^2\mathrm{d}\zeta & iu(s) \cos \varphi_0 \\ -iu(s) \cos \varphi_0 & -i\int_s^\infty u(\zeta)^2\mathrm{d}\zeta \end{pmatrix}+  \mathcal{O}\left(t^{-1/3-5\tau}\right), \label{e001}\\[6pt]
		& E(0) = I+(3t)^{-1/3} \begin{pmatrix}0& u(s) \sin \varphi_0 \\ u(s) \sin \varphi_0 & 0 \end{pmatrix}+ \mathcal{O}\left(t^{-1/3-5\tau}\right).\label{e00}
	\end{align}
\end{proposition}

\begin{proof}
By   (\ref{ioei}) and  (\ref{owej}), we obtain that
	\begin{align*}
E_1 & = -\frac{1}{2\pi i} \oint_{\partial \mathcal{U}_{r}} \left( V^{E}(\zeta)-I \right) \mathrm{d}\zeta -\frac{1}{2\pi i} \oint_{\partial \mathcal{U}_{l}} \left( V^{E}(\zeta)-I \right) \mathrm{d}\zeta
+ \mathcal{O}\left(t^{-1/3-5\tau}\right)\\
		& = (3t)^{-1/3} \left(M^{r}_1(s)+M^{l}_1(s)\right)+ \mathcal{O}\left(t^{-1/3-5\tau}\right),
	\end{align*}
which gives (\ref{e001}) by  the estimates (\ref{mr1}) and (\ref{ml1}). Here, we use the fact $0<\tau<\frac{1}{30}$.

	In a similar way, we have
	\begin{align*}
		E(0) &= I + \frac{1}{2\pi i} \oint_{\partial \mathcal{U}_{r}}  \frac{ V^{E}(\zeta)-I}{\zeta}  \mathrm{d}\zeta +\frac{1}{2\pi i} \oint_{\partial \mathcal{U}_{l}}\frac{ V^{E}(\zeta)-I}{\zeta}  \mathrm{d}\zeta + \mathcal{O}\left(t^{-1/3-5\tau}\right)\\
		&= I -(3t)^{-1/3 } \left(M^{r}_1(s)-M^{l}_1(s)\right) + \mathcal{O}\left(t^{-1/3-5\tau}\right),
	\end{align*}
which yields   (\ref{e00}) by  the estimates (\ref{mr1}) and (\ref{ml1}).
\end{proof}

\subsection{Contribution from a pure $\bar{\partial}$-problem}\label{modi4}
In this subsection, we   consider the long-time   asymptotic behavior of  the pure $\bar{\partial}$-problem.
   Define the function
            \begin{equation}\label{transd}
                M^{(5)}(z)=M^{(4)}(z)\left(M^{rhp}(z)\right)^{-1},
            \end{equation}
            which  satisfies the following $\bar{\partial}$-problem.

\begin{prob3}\label{trad}
 Find  $M^{(5)}(z):=M^{(5)}(z;x,t)$ which satisfies
            \begin{itemize}
                \item   $M^{(5)}(z)$ is continuous  in $\mathbb{C}$ and analytic in $\mathbb{C}\setminus \overline{\Omega}$.
                \item    $M^{(5)}(z)=I+\mathcal{O}(z^{-1}), \ z\to  \infty$.	

                \item For $z\in \mathbb{C}$,  $M^{(5)}(z)$ satisfies the $ \bar{\partial}$-equation
                \begin{equation}
                    \bar{\partial} M^{(5)}(z) = M^{(5)}(z) W^{(5)}(z), \nonumber
                \end{equation}
                where
	\begin{align}
	W^{(5)}(z):=M^{rhp}(z)  \bar{\partial} R^{(3)}(z) \left(M^{rhp}(z)\right)^{-1}, \label{633s}
	\end{align}
 and  $ \bar{\partial}  R^{(3)}(z)$ has been given in (\ref{parR2}).
            \end{itemize}

\end{prob3}

The solution of $ \bar{\partial}$-problem   \ref{trad} can be given by
\begin{equation} \label{Im3}
	M^{(5)}(z)=I-\frac{1}{\pi}  \iint_\mathbb{C} \frac{ M^{(5)}(\zeta) W^{(5)}(\zeta)}{\zeta-z} \, \mathrm{d} A(\zeta),
\end{equation}
where $\mathrm{d} A(\zeta)$ is Lebesgue measure on the plane.
\eqref{Im3} can be  written as an operator equation
\begin{equation}
	(I-S) M^{(5)}(z)=I, \label{Sm3}
\end{equation}
where $S$ is the solid Cauchy operator
\begin{equation}
	Sf(z)=-\frac{1}{\pi} \iint_\mathbb{C} \frac{f(\zeta)W^{(5)}(\zeta)}{\zeta-z} \, \mathrm{d} A(\zeta). \label{hfu}
\end{equation}

\begin{proposition}\label{pss}
The    operator $S$ defined by (\ref{hfu})  satisfies the estimate
	\begin{equation}\label{esuf}
  \parallel S\parallel_{L^\infty(\mathbb{C})\to L^\infty(\mathbb{C})}\lesssim  t^{-1/6},
	\end{equation}
	which implies the existence of $(I-S)^{-1}$ for large $t$.
\end{proposition}
\begin{proof}
 We   estimate   the operator $S$  on  $\Omega_{1}$  and other cases  are   similar.  Following the argument in Lemma 6.11 \cite{CJ},  for a fixed constant  $c$, we have
 \begin{equation*}
     |Sf(z)| \le c ||f||_{L^\infty(\mathbb{C})} \int_{\Omega_1} \frac{\langle \zeta \rangle|\bar\partial R_1(\zeta) e^{2it\theta(\zeta)}|}{|\zeta-z||\zeta-1|}\, \mathrm{d} A(\zeta).
 \end{equation*}
 In fact,  to prove \eqref{esuf}, using (\ref{437}), (\ref{estn1}), (\ref{R3}),  (\ref{633s}), and (\ref{hfu}),
 it is sufficient to show that
	\begin{align*}
	|Sf|  \leq c (I_1+I_2+I_3+I_4),
	\end{align*}
where
	\begin{align*}
		&I_1  =    \iint_{\Omega_{1} \cap \{\zeta:|\zeta| \le 2\} } F(\zeta,z) \mathrm{d} A(\zeta), \ \  I_2 =  \iint_{\Omega_{1} \cap \{\zeta:|\zeta| > 2\}}  F(\zeta,z) \mathrm{d} A(\zeta),\\
&I_3  =    \iint_{\Omega_{1} \cap \{\zeta:|\zeta| \le 2\} } G(\zeta,z)\mathrm{d} A(\zeta), \  \ \ I_4=\iint_{\Omega_{1} \cap \{\zeta:|\zeta| > 2\}} G(\zeta,z)	\mathrm{d} A(\zeta),
	\end{align*}
with
$$  F(\zeta,z)= \frac{1}{|\zeta-z|}\left| f(|\zeta|)\right|e^{\re(2it\theta(z))}, \  G(\zeta,z)= \frac{1}{|\zeta-z|}\left| \zeta-z_1 \right|^{-1/2}e^{\re(2it\theta(z))}.$$
Here, $f(|\zeta|) =r'(|\zeta|)$ or $f(|\zeta|) =\varphi( |z|)$.
Let $y= \im z$ and $\zeta =z_1 + u+iv = |\zeta|e^{iw}$.
   Using Proposition \ref{reprop1} and  the Cauchy-Schwarz inequality, we have
	\begin{align*}
		 I_1  &\le \int_{0}^{2 \sin w} \int_{v}^{2\cos w-z_1} F(\zeta,z) \mathrm{d}u \mathrm{d}v 	\lesssim \int_{0}^{2 \sin w} t^{-\frac{1}{4}} |v-y|^{-\frac{1}{2}} v^{-\frac{1}{4}}e^{-c_1 tv^3} \mathrm{d}v   \lesssim t^{-1/3},\\
         I_2  &\le \int_{2 \sin w}^\infty \int_{2\cos w -z_1}^{\infty} F(\zeta,z) \mathrm{d}u \mathrm{d}v 	\lesssim \int_{2 \sin w}^\infty ||r'||_{L^2(\mathbb{R})}   |v-y|^{-1/2} e^{-c_1 tv} \mathrm{d}v   \lesssim t^{-1/2}.
	\end{align*}

In a similar way, using Proposition \ref{reprop1} and   the H\"{o}lder's inequality with $p>2$ and $1/p+1/q=1$, we obtain
	\begin{align*}
		 I_3  &\lesssim  \int_{0}^{2 \sin w} v^{1/p-1/2}|v-y|^{1/q-1}e^{-c_1 tv^3} \mathrm{d}v \lesssim t^{-1/6},\\
       I_4  &\lesssim  \int_{2 \sin w}^\infty v^{1/p-1/2}|v-y|^{1/q-1}e^{-c_1 tv} \mathrm{d}v \lesssim t^{-1/2}.
	\end{align*}
\end{proof}

 Proposition \ref{pss}  implies that   the operator  equation (\ref{Sm3})  exists a unique solution, which
can be expanded in the form
\begin{equation}
	M^{(5)}(z)=I + z^{-1} M^{(5)}_1(x,t) +\mathcal{O}(z^{-2}), \quad z\to \infty, \label{trans9}
\end{equation}
where
\begin{align}\label{expanm51}
M^{(5)}_1(x,t)=\frac{1}{\pi} \iint_{\mathbb{C}} M^{(5)}(\zeta)W^{(5)}(\zeta)\, \mathrm{d} A(\zeta).
\end{align}

Take $z=0$ in (\ref{Im3}), then
\begin{equation} \label{m50}
	M^{(5)}(0)=I-\frac{1}{\pi}  \iint_\mathbb{C} \frac{ M^{(5)}(\zeta) W^{(5)}(\zeta)}{\zeta}  \mathrm{d} A(\zeta).
\end{equation}
\begin{proposition}   We have the  following estimates
	\begin{equation}\label{m51infty}
		|M^{(5)}_1(x,t)| \lesssim t^{-1/2}, \ \ \ |M^{(5)}(0)-I| \lesssim t^{-1/2}.
	\end{equation}
\end{proposition}
\begin{proof}
	Similarly to the proof of Proposition \ref{pss}, we take $z \in \Omega_{1}$ as an example and divide the integration (\ref{expanm51}) on $\Omega_{1}$ into four parts.
Firstly, we consider the estimate of  $M^{(5)}_1(x,t)$.
By (\ref{transd}) and the boundedness of $M^{(4)}(z)$ and $M^{rhp}(z)$ on $\Omega_{1}$, we have
	\begin{align}\label{m51esti}
	|M^{(5)}_1(x,t)| \lesssim I_1+I_2+I_3+I_4,
	\end{align}
	where
	\begin{align*}
		&I_1  =    \iint_{\Omega_{1} \cap \{\zeta:|\zeta| \le 2\} } F(\zeta,z) \mathrm{d} A(\zeta), \ \  I_2 =  \iint_{\Omega_{1} \cap \{\zeta:|\zeta| > 2\}}  F(\zeta,z) \mathrm{d} A(\zeta),\\
&I_3  =    \iint_{\Omega_{1} \cap \{\zeta:|\zeta| \le 2\} } G(\zeta,z)\mathrm{d} A(\zeta), \  \ \ I_4=\iint_{\Omega_{1} \cap \{\zeta:|\zeta| > 2\}} G(\zeta,z)	\mathrm{d} A(\zeta),
	\end{align*}
with
$$  F(\zeta,z)= \left|f(|\zeta|) \right|e^{\re(2it\theta(z))}, \  G(\zeta,z)= \left| \zeta-z_1 \right|^{-1/2}e^{\re(2it\theta(z))}.   $$
Here, $f(|\zeta|) =r'(|\zeta|)$ or $f(|\zeta|) =\varphi( |z|)$.
 Let  $\zeta =z_1 + u+iv = |\zeta|e^{iw}$.  By Cauchy-Schwarz inequality and Proposition \ref{reprop1}, we have
	\begin{align*}
		I_1 & \lesssim \int_{0}^{2 \sin w} \int_{v}^{2 \cos w -z_1} \left|f(|\zeta|) \right|e^{-c_1 tu^2v} \mathrm{d}u \mathrm{d}v  \lesssim t^{-1/2},\\
I_2 & \lesssim \int_{2 \sin w}^\infty \int_{2 \cos w -z_1}^{\infty} \left|f(|\zeta|) \right| e^{-c_1 tv} \mathrm{d}u \mathrm{d}v  \lesssim t^{-1}.
	\end{align*}
	By H\"{o}lder's inequality with $p>2$ and $1/p+1/q=1$ and Proposition \ref{reprop1}, we have
	\begin{align*}
		I_3 &\lesssim \int_{0}^{2 \sin w} \int_{v}^{2\cos w -z_1}   \left| u + i v\right|^{-1/2} e^{-c_1 tu^2v} \mathrm{d}u  \mathrm{d}v \lesssim t^{-1/2-1/(3p)},\\
I_4 &\lesssim \int_{2 \sin w}^\infty \int_{2 \cos w -z_1}^{\infty}  \left| u + i v\right|^{-1/2}e^{-c_1 tv} \mathrm{d}u  \mathrm{d}v \lesssim t^{-3/2}.
	\end{align*}

To estimate $M^{(5)}(0)-I$,  we first note that $|z|^{-1} \le |z_1|^{-1}$
for all $z \in \Omega_{1}$. Combining the estimates from (\ref{m50}) and (\ref{m51esti}), we obtain $|M^{(5)}(0)-I| \lesssim t^{-1/2}$.

\end{proof}

 To recover the potential via  the reconstruction formula (\ref{sol}), we require an estimate of  $M^{(3)}(0)$.
\begin{proposition} As $t \to +\infty$, $M^{(3)}(0)$ satisfies the   estimate
	\begin{align}\label{m30}
		M^{(3)}(0)= E(0) + \mathcal{O}(t^{-1/2}),
	\end{align}
where $E(0)$ is given by (\ref{e00}).
\end{proposition}
\begin{proof}
	Reviewing the series of transformations  (\ref{trans4}), (\ref{trans6}), and (\ref{transd}) ,  for large $z$ and satisfying   $R^{(3)}(z)=I$,  the solution of $M^{(3)}(z)$ is given by
	\begin{align*}
		M^{(3)}(z) = M^{(5)}(z)E(z) .
	\end{align*}
By (\ref{trans9}) and (\ref{m51infty}), we further obtain
\begin{align*}
	M^{(3)}(z)= E(z)  + \mathcal{O} (t^{-1/2}),
\end{align*}
which yields (\ref{m30}) by taking $z=0$.
\end{proof}

\subsection{Painlev\'e asymptotics} \label{thm-result1}


In this subsection,  we  state  and prove the  main result   in this paper  as follows.

\begin{theorem} \label{th}  For initial data $q_0(x)-\tanh{(x)}\in H^{4,4}(\mathbb{R})$,
let $r(z)$ and  $\{\eta_j\}_{j=0}^{N-1} $  be the  associated
 reflection coefficient and the discrete spectrum, respectively.
 We also define  the  modified reflection coefficient  $\tilde{r}(z)$  given by (\ref{Rz2})-(\ref{Rz22}).
 Then the long-time asymptotics of the solution to the Cauchy problem (\ref{dmkdv})-(\ref{bdries})  for the
 defocusing mKdV  equation in the transition region $\left|\frac{x}{t}+6 \right|t^{2/3}< C$ with $C>0$  is given by the following formula:

 \begin{equation}\label{q1}
  q(x,t)= -1 +(3t)^{-1/3}u(s) \cos \varphi_0 +\mathcal{O}\left(t^{-1/3-5\tau}\right),
 \end{equation}
where $\tau$ is a constant with $0<\tau<1/30$,
\begin{align*}
 s=\frac{1}{3} (3t)^{2/3} \left(\frac{x}{t}+6\right),\ \varphi_0 = \arg \tilde{r}(1),
\end{align*}
and $u(s ) $ is a solution of the Painlev\'{e} \uppercase\expandafter{\romannumeral2} equation
\begin{equation}\label{pain2equ}
u_{ss}(s)  = 2u^3(s) + s u(s),
\end{equation}
which admits  the asymptotics
\begin{align}\label{pain2equa}
 u(s) \sim -|\tilde{r}(1)| \mathrm{Ai}(s) \sim  - |\tilde{r}(1)| \frac{1}{2\sqrt{\pi}}s^{-1/4}  e^{-2 s^{3/2} /3}, \ s \to +\infty,
\end{align}
 where $\mathrm{Ai}(s)$ is the classical Airy function.
\end{theorem}

\begin{proof}
	Inverting  the sequence  of transformations (\ref{trans1}), (\ref{trans3}), (\ref{trans4}),  (\ref{trans6}), (\ref{transd}),
	and especially taking $z\to \infty$ vertically such that  $R^{(3)}(z)= G(z)=I$, then  the  solution of RH problem \ref{RHP0} is given by
	\begin{align}
		& 	M(z) =\left(  I + \frac{\sigma_2}{z}M^{(3)}(0)^{-1} \right)  M^{(5)}(z) E(z)\delta(z)^{-\sigma_3} + \mathcal{O}(e^{-ct}), \nonumber
	\end{align}
    where $c$ is a positive constant.
	Furthermore, substituting asymptotic expansions  (\ref{Texpan}),  (\ref{trans8}),   and (\ref{trans9}) into  the above formula,
	the reconstruction formula (\ref{sol}) yields
	\begin{equation}\label{qori}
		q(x,t) =i  \left(\sigma_2 M^{(3)}(0)^{-1} + E_1 \right)_{21} + \mathcal{O}\left(t^{-1/3-5\tau}\right).
	\end{equation}
 Utilizing  (\ref{e001}) and  (\ref{m30}), we arrive at   the result stated  as (\ref{q1})  in Theorem \ref{th}.
\end{proof}


\appendix

\section{ Painlev\'{e} \uppercase\expandafter{\romannumeral2}   Model RH Problem} \label{appx}
The Painlev\'{e} \uppercase\expandafter{\romannumeral2} equation takes the form
\begin{align}\label{up22}
	u_{ss}(s) = 2u^3(s)  +su(s), \quad s \in \mathbb{R},
\end{align}
which can be solved by means of the solution of a  RH problem as follows.
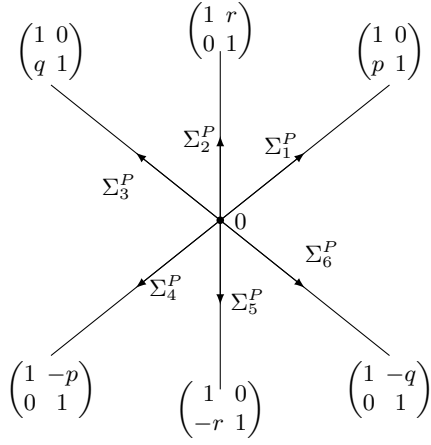
\begin{figure}
	\begin{center}
		\begin{tikzpicture}[scale=0.9]
			\node[shape=circle,fill=black,scale=0.15] at (0,0) {0};
			\node[below] at (0.3,0.25) {\footnotesize $0$};
			\draw [] (0,-2.5 )--(0,2.5);
			\draw [-latex] (0,0)--(0,1.25);
			\draw [-latex] (0,0 )--(0,-1.25);
			\draw [] (0,0 )--(2.5,2);
			\draw [-latex] (0,0)--(1.25,1);
			\draw [] (0,0 )--(2.5,-2);
			\draw [-latex] (0,0)--(1.25,-1);
			\draw [] (0,0 )--(-2.5,2);
			\draw [-latex] (0,0)--(-1.25,1);
			\draw [] (0,0 )--(-2.5,-2);
			\draw [-latex] (0,0)--(-1.25,-1);

			\node at (0.9,1.1) {\footnotesize$\Sigma^P_1$};
			\node at (1.5,-0.5) {\footnotesize$\Sigma^P_6 $};
			\node at (-1.5,0.5) {\footnotesize$\Sigma^P_3$};
			\node at (-0.8,-1) {\footnotesize$\Sigma^P_4$};
			\node at (-0.3,1.2) {\footnotesize$\Sigma^P_2$};
			\node at ( 0.4,-1.2) {\footnotesize$\Sigma^P_5$};
			
			\node  at (2.5,2.5) {\footnotesize $\begin{pmatrix} 1 & 0 \\ p & 1 \end{pmatrix}$};
			\node  at (-2.5,2.5) {\footnotesize $\begin{pmatrix} 1 & 0\\ q& 1 \end{pmatrix}$};
			\node  at (0,2.8) {\footnotesize $\begin{pmatrix} 1 & r\\ 0& 1 \end{pmatrix}$};
			\node  at (0,-2.8) {\footnotesize $\begin{pmatrix} 1 & 0 \\ -r& 1 \end{pmatrix}$};
			\node  at (2.5,-2.5) {\footnotesize $\begin{pmatrix} 1 & -q\\ 0& 1 \end{pmatrix}$};
			\node  at (-2.5,-2.5) {\footnotesize $\begin{pmatrix} 1 & -p \\ 0 & 1 \end{pmatrix}$};
			
		\end{tikzpicture}
		\caption{ \footnotesize { The jump contour $\Sigma^P$ and the corresponding jump matrix.}}
		\label{Sixrays}
	\end{center}
\end{figure}

Denote $\Sigma^P = \bigcup_{j=1}^6\left\{  \Sigma_j^P = e^{i\left(\frac{\pi}{6}+(j-1)\frac{\pi}{3}\right)} \mathbb{R}_+ \right\}$, see Figure \ref{Sixrays}.
The Painlev\'{e} \uppercase\expandafter{\romannumeral2} model RH problem satisfies the following properties:
\begin{prob}\label{mpain}
	Find   $M^{P}(k)=M^{P}(k;s)$ with properties
	\begin{itemize}
		\item Analyticity: $M^{P}(k)$ is analytic in $\mathbb{C}\setminus \Sigma^{P}$.
		\item Jump condition:
		\begin{equation*}
			M^{P}_+( k)=M^{P}_-(k)e^{-i(\frac{4}{3}k^3+sk)\widehat{\sigma}_3}V^{P}(k),
		\end{equation*}
		where
		$V^{P}(k)$ is shown in Figure \ref{Sixrays}. The parameters $p$, $q$, and $r$ in $V^{P}(k)$   satisfy the relation
		\begin{align}\label{rer}
		r=	p+q+pqr .
		\end{align}
		\item Asymptotic behaviors:
		\begin{align*}
			& M^{P}(k) = I +   \mathcal{O} \left(k^{-1}\right), \quad \text{as} \ k \to \infty,\\
			& M^{P}( k) = \mathcal{O}(1),\quad  \text{as} \ k \to 0,
		\end{align*}
and for each $C_1 > 0$,
\begin{align}\label{mPbounded}
	\sup_{k \in \mathbb{C}\setminus \Sigma^P} \sup_{s \geq -C_1} |M^P(k)|  < \infty.
\end{align}
	\end{itemize}
\end{prob}

Then
\begin{equation}\label{up2}
	u(s)=2\left(M_1^P(s)\right)_{12} = 2 \left(M_1^P(s)\right)_{21}
\end{equation}
solves the Painlev\'{e} \uppercase\expandafter{\romannumeral2} equation, where
	\begin{align*}
			& M^{P}(k) = I + k^{-1} M_1^P(s) + \mathcal{O} \left(k^{-2}\right), \quad \text{as} \ k \to \infty.
		\end{align*}
A result due to Hastings and McLeod \cite{p1} presents that, for any $a\in\mathbb{R}$, there exist a unique solution to the homogeneous
Painlev\'{e} II equation \eqref{up22} that behaves like
\begin{align}\label{us}
	u(s) = a \mathrm{Ai}(s) +\mathcal{O}\left( s^{-\frac{1}{4}}e^{-\frac{4}{3}s^{3/2}}\right),\quad s\to +\infty,
\end{align}
where  $\mathrm{Ai}(s)$ denotes the Airy function. Particularly, for $q\in i\mathbb{R}$, $|q|<1$, $p=-q$,  and $r=0$, it follows that the solution $u(s)$ has the asymptotics \eqref{us} with
$a=-\im q$  and
the matrix  $M_1^P(s)$ has the form given by
\begin{align}\label{posee}
M_1^P(s) = \frac{1}{2} \begin{pmatrix} -i\int_s^\infty u(\zeta)^2\mathrm{d}\zeta & u(s) \\ u(s) & i\int_s^\infty u(\zeta)^2\mathrm{d}\zeta \end{pmatrix}.
\end{align}
Furthermore, a special argument shows that 
for the singular case $q\in i \mathbb{R}, |q|=1$, $p=-q$, and $r=0$, \eqref{up2} also leads to a global, real solution of \eqref{up22} with the asymptotics  \eqref{us}. 
More details can be found in \cite{AblwzP1,as1,p1,p2}.



 \vspace{4mm}

    \noindent\textbf{Acknowledgments.}
    Wang is supported by the National Natural Science
    Foundation of China (Grant No. 12347141) and China
    Postdoctoral Science Foundation (Certificate No. 2023M740717).
    Xu is supported by China
    Postdoctoral Science Foundation (Certificate No. 2024M760480).
    Fan is supported by the National Natural Science
    Foundation of China (Grant No. 12271104).
    \vspace{2mm}

    \noindent\textbf{Data Availability Statements}

    The data that supports the findings of this study are available within the article.\vspace{2mm}

    \noindent{\bf Conflict of Interest}

    The authors have no conflicts to disclose.


\begin{thebibliography}{99}
  \bibitem{Zab1967}N. Zabusky,
  \newblock{Proceedings of the Symposium on nonlinear partial differential equations},
  Academic Press Inc., New York, 1967.

  \bibitem{Ono1992}H. Ono,
  \newblock{Soliton fission in anharmonic lattices with reflectionless inhomogeneity},
  J. Phys. Soc. Jpn., 61 (1992), 4336-4343.

  \bibitem{kaku1969}T. Kakutani, H. Ono,
  \newblock{Weak non-linear hydromagnetic waves in a cold collision-free plasma},
  J. Phys. Soc. Jpn., 26 (1969), 1305-1318.

  \bibitem{Kha1998}A. Khater, O. El-Kalaawy, D. Callebaut,
  \newblock{B\"acklund transformations and exact solutions for Alfv\'en solitons in a relativistic electron-positron plasma},
  Phys. Scr., 58 (1998), 545.


  \bibitem{mkdv1}
  C. Kenig, G. Ponce, L. Vega,
  \newblock{Well-posedness and scattering results for the generalized Korteweg-de
Vries equation via the contraction principle},
  Commun. Pure Appl. Math., 46 (1993), 527-620.

  \bibitem{mkdv2}
  J. Colliander, M. Keel, G. Staffilani, H. Takaoka, T. Tao,
   \newblock{Sharp global well-posedness for KdV and modified KdV on $\mathbb{R}$ and $\mathbb{T}$},
   J. Am. Math. Soc., 16 (2003),  705-749.


  \bibitem{mkdv3}
  Z. Guo,
  \newblock{Global well-posedness of Korteweg-de Vries equation in $H^{-3/4}(\mathbb{R})$},
  J. Math. Pure Appl., 91 (2009), 583-597.


  \bibitem{mkdv4}
  N. Kishimoto,
   \newblock{Well-posedness of the Cauchy problem for the Korteweg-de Vries equation at the critical regularity},
    Differ. Integral Equ., 22 (2009), 447-464.



  \bibitem{mkdv6}
  B. Harrop-Griffiths, R. Killip, M. Visan,
  \newblock{Sharp well-posedness for the cubic NLS and mKdV in $H^s(\mathbb{R})$},
Forum Math. Pi, 12 (2024), e6\; 1-86.




  \bibitem{DeiftP2}P. Deift, X. Zhou,
  \newblock{A steepest descent method for oscillatory Riemann--Hilbert problems. Asymptotics for the MKdV equation},
  Ann. Math., 137 (1993), 295-368.

    \bibitem{AblwzP1}H. Segur, M. J. Ablowitz,
  \newblock{Asymptotic solutions of nonlinear evolution equations and a Painlev\'{e} transcendent},
  Physica D, 3 (1981), 165-184.



 \bibitem{MK0}
 A. Minakov,
  	Long-time behavior of the solution to the mKdV	equation with step-like initial data,
    J. Phys. A: Math. Theor.,  44 (2011), 085206.


  	\bibitem{MK1}
  V. P. Kotlyarov, A. Minakov,
  	Step-initial function to the mKdV equation: Hyper-Elliptic long-time asymptotics of the solution,
  J. Math. Phys.,  8 (2012), 38-62.


    \bibitem{GM}
T. Grava, A. Minakov,
	On the long-time asymptotic behavior of the modified Korteweg-de Vries equation with step-like initial data,
   SIAM J. Math. Anal.,  52 (2020), 5892-5993. 



	   \bibitem{VA1}
V. P. Kotlyarov, A. Minakov,
	\newblock Riemann-Hilbert problem to the modified Korteveg-de Vries equation: Long-time dynamics of the steplike initial data,
  J. Math. Phys.,  51 (2010), 093506.


  	   \bibitem{VA1ST}
 V. P. Kotlyarov, A. Minakov,
  	\newblock Riemann-Hilbert problems and the mKdV equation
  	with step initial data: Short-time behavior of solutions
  	and the nonlinear Gibbs-type phenomenon,
  J. Phys. A: Math. Theor.,  45 (2012), 325201. 

\bibitem{Monvel}
A. Boutet de Monvel, D. Shepelsky,
\newblock Initial boundary value problem for the mKdV equation on a finite interval,
\newblock { Ann. Inst. Fourier}, 54 (2004),  1477-1495. 

  	   \bibitem{chenliu}  G. Chen, J. Q. Liu, Long time asymptotics  of   the   modified
KdV   equation   in  weighted Sobolev spaces, Forum Math. Sigma, 10 (2022), e66\, 1-52.





   \bibitem{Charlier2020}
    C. Charlier, J. Lenells, Airy and Painlev\'{e} asymptotics for the mKdV equation,
 J.  Lond.  Math.  Soc.,   101  (2020),  194-225.

    \bibitem{Huang2020}
     L. Huang, L. Zhang, Higher order Airy and Painlev\'e asymptotics for the mKdV hierarchy,
      SIAM J. Math. Anal., 54 (2022), 5291-5334.


\bibitem{Its3}    A. Its,  A. Prokhorov,  Connection problem for the tau-function of the sine-Gordon
  reduction of Painlev\'e-III equation via the Riemann-Hilbert  approach, Int. Math. Res. Not.,  375 (2016),  6856-6883.



  \bibitem{Monvel2010}
   A. Boutet de Monvel, A. Its,  D. Shepelsky, Painlev\'{e}-type asymptotics for the Camassa-Holm  equation,
    SIAM J.  Math.  Anal.,  42  (2010), 1854-1873.

   \bibitem{wfp}
  Z. Y. Wang, E. G. Fan,
  \newblock{The defocusing nonlinear Schr\"{o}dinger equation with a nonzero background: Painlev\'{e} asymptotics in two transition regions},
Commun. Math. Phys., 402 (2023), 2879-2930.



\bibitem{xyz}
T. Y. Xu, Y. L. Yang, and L. Zhang,
  \newblock{ Transient asymptotics of the modified Camassa-Holm equation}, J. Lond. Math. Soc., 110 (2024), e12967.

   \bibitem{Z&Y} G. Q. Zhang, Z. Y. Yan,
    \newblock{Focusing and defocusing mKdV equations with nonzero boundary conditions: Inverse scattering transforms and soliton interactions},
    \newblock{Physica D}, 410 (2020), 132521.




\bibitem{McL1}
    K. T. R. McLaughlin, P. D. Miller,
    \newblock The $\bar{\partial}$ steepest descent method and the asymptotic behavior of polynomials orthogonal on the unit circle with fixed and exponentially varying non-analytic weights,
    \newblock {Int. Math. Res. Not.}, 2006 (2006), 48673.

  \bibitem{McL2}
    K. T. R. McLaughlin, P. D. Miller,
    \newblock The $\bar{\partial}$ steepest descent method for orthogonal polynomials on the real line with varying weights,
    \newblock {Int. Math. Res. Not.}, 2008 (2008), 075.



 \bibitem{fNLS}
M. Borghese, R. Jenkins, K. T. R. McLaughlin,  P. D. Miller,
\newblock { Long-time asymptotic behavior of the focusing nonlinear Schr\"odinger equation, }
\newblock {  Ann. I. H. Poincar\'e-Anal.}, 35 (2018), 887-920.



\bibitem{Liu3}
R. Jenkins, J. Liu, P. Perry, C. Sulem,
\newblock Soliton resolution for the derivative nonlinear Schr\"odinger equation,
\newblock {Commun. Math. Phys.},  363 (2018), 1003-1049.



\bibitem{LJQ}J. Q. Liu,
\newblock Long-time behavior of solutions to the derivative nonlinear Schr\"odinger equation for soliton-free initial data,
\newblock {Ann. I. H. Poincar\'e -Anal.}, 35 (2018), 217-265.

\bibitem{YF1} Y. L. Yang, E. G. Fan, \newblock   On the long-time asymptotics of the modified Camassa-Holm equation in space-time solitonic regions,  Adv. Math., 402 (2022), 108340.



\bibitem{WF}   Z. Y. Wang,  E. G. Fan,  The defocusing NLS equation with nonzero background: Large-time asymptotics in the solitonless region,   J. Differential Equations, 336 (2022), 334-373.

  \bibitem{zx1}
  Z. C. Zhang, T. Y. Xu, E. G. Fan,
  \newblock{Soliton resolution and asymptotic stability of $N$-soliton solutions for the defocusing mKdV equation with finite density type initial data}, Physica D, 472 (2025), 134526.




  \bibitem{zx2}
  T. Y. Xu, Z. C. Zhang, E. G. Fan,
  \newblock{On the Cauchy problem of defocusing mKdV equation with finite density initial data: Long time asymptotics in soliton-less regions},
  J. Differential Equations, 372 (2023), 55-122.

   \bibitem{CJ}
   S. Cuccagna, R. Jenkins, On the asymptotic stability of $N$-soliton solutions of the defocusing nonlinear Schr\"odinger equation,
   \newblock  Commun. Math. Phys., 343  (2016), 921-969.



    \bibitem{PX3}
P. Deift, X. Zhou,
\newblock Long-time asymptotics for solutions of the NLS equation with initial data in a weighted Sobolev space,
\newblock Commun. Pure Appl. Math., 56 (2003), 1029-1077.



\bibitem{Fokas1}  A. Fokas, A. Its, A. Kapaev, V. Novokshenov,  Painlev\'e  Transcendents: The Riemann-Hilbert Approach,  American Mathematical Society Mathematical Surveys and Monographs, 128, Providence, RI: American Mathematical Society, 2006.



  \bibitem{BC84}
  R. Beals, R. R. Coifman,
  \newblock  Scattering and inverse scattering for first order systems,
  \newblock Commun. Pure Appl. Math., 37 (1984), 39-90.


\bibitem{as1}
 H. Segur,  M. J. Ablowitz,  Asymptotic solutions of the Korteweg-de Vries equation, Stud.
Appl. Math., 571 (1977), 13-44.




\bibitem{p1}
 S. Hastings, J. B. McLeod,  A boundary value problem associated with the second Painlev\'e transcendent and the Korteweg-de Vries equation, Arch. Ration. Mech. Anal., 73 (1980), 31-51.


\bibitem{p2}
P. A. Deift, X. Zhou, {Asymptotics for the {P}ainlev{\'e} {I}{I}
equation}, Commun. Pure Appl. Math., 48 (1995),  277-337.
	
	
	
	
	
\end{thebibliography}
\end{document}